\documentclass{CSML}
\pdfoutput=1

\usepackage{lastpage}
\newcommand{\dOi}
{14(2:19)2018}
\lmcsheading{}{1--\pageref{LastPage}}{}{}%
{Jan.~31, 2017}
{Jun.~29, 2018}
{}


\usepackage{hyperref}
\hypersetup{hidelinks}
\usepackage{amsmath}
\usepackage{stmaryrd}
\usepackage{latexsym}

\newcommand{\hide}[1]{}

\newcommand*{\pnot}{\mathord{\sim}}

\newcommand{\mo}[1]{\llbracket#1\rrbracket}

\newcommand{\mwrs}[3]{\llbracket#1\rrbracket_{#3}\left(#2\right)}

\newcommand{\vgr}[2]{#2(#1)}

\newcommand{\vwrt}[2]{\mathit{val}_{#2}(#1)}

\newcommand{\aleq}[1][]{\sqsubseteq_{#1}}
\newcommand{\ale}[1][]{\sqsubset_{#1}}

\newcommand{\mgp}{M_\mathsf{Gr(P)}}
\newcommand{\wfmp}{\mathcal{M}_\mathsf{P}}
\newcommand{\wgp}{M_\mathsf{Gr(P)}}
\newcommand{\baseap}[1]{I_{#1}}

\newcommand{\exeq}[1][]{\cong_{#1}}

\newcommand{\unk}{0}

\makeatletter
\renewcommand\paragraph{\@startsection{paragraph}{4}{\parindent}{\z@}%
                                    {-\fontdimen2\font}%
                                    {\normalfont\normalsize\bfseries}}
\makeatother

\begin{document}
\title{Extensional Semantics for Higher-Order Logic Programs with Negation{\rsuper*}}
\titlecomment{{\lsuper*}\kern-1ptA preliminary version of this paper has appeared in the
proceedings of the 15th European Conference on Logics in
Artificial Intelligence (JELIA), pages 447--462, 2016.}

\author[P.~Rondogiannis]{Panos Rondogiannis}	
\author[I.~Symeonidou]{Ioanna Symeonidou}	
\address{Department of Informatics \& Telecommunications,
National and Kapodistrian University of Athens, Greece}	
\email{\{prondo,i.symeonidou\}@di.uoa.gr}  

\begin{abstract}
We develop an extensional semantics for higher-order logic programs with negation,
generalizing the technique that was introduced in~\cite{Bezem99,Bezem01} for positive
higher-order programs. In this way we provide an alternative extensional semantics for
higher-order logic programs with negation to the one proposed in~\cite{CharalambidisER14}.
We define for the language we consider the notions of {\em stratification} and {\em local stratification}, 
which generalize the familiar such notions from classical logic programming, and we demonstrate that for 
stratified and locally stratified higher-order logic programs, the proposed semantics never assigns the
{\em unknown} truth value. We conclude the paper by providing a negative result: we demonstrate
that the well-known stable model semantics of classical logic programming, if extended
according to the technique of~\cite{Bezem99,Bezem01} to higher-order logic programs,
does not in general lead to extensional stable models.

\end{abstract}
\maketitle              

\section{Introduction}
Research results developed in~\cite{Wa91a,Bezem99,Bezem01,KRW05,CharalambidisHRW13} have
explored the possibility of designing higher-order logic programming languages with
{\em extensional semantics}. Extensionality implies that program predicates essentially
denote sets, and therefore one can use standard set theoretic concepts in order to understand
the meaning of programs and reason about them. The key idea behind this line of research
is that if we appropriately restrict the syntax of higher-order logic programming, then we
can achieve extensionality and obtain languages that are simple both from a semantic as-well-as
from a proof-theoretic point of view. Therefore, a main difference between the extensional and the more
traditional {\em intensional} higher-order logic programming languages~\cite{MN2012,CKW93-187}
is that the latter have richer syntax and expressive capabilities but a non-extensional semantics.

There exist at present two main extensional semantic approaches for capturing the meaning of
positive (i.e., negationless) higher-order logic programs. The first approach, developed
in~\cite{Wa91a,KRW05,CharalambidisHRW13}, uses classical domain-theoretic tools. The second
approach, developed in~\cite{Bezem99,Bezem01}, uses a fixed-point construction on the
ground instantiation of the source program. Despite their different philosophies, these two approaches
have recently been shown~\cite{CharalambidisRS15} to agree for a broad and useful class of
programs. This fact suggests that the two aforementioned techniques can be employed
as useful alternatives for the further development of higher-order logic programming.

A natural question that arises is whether one can still obtain an extensional semantics
if negation is added to programs. This question was recently undertaken
in~\cite{CharalambidisER14}, where it was demonstrated that the domain-theoretic
results obtained for positive logic programs in~\cite{Wa91a,KRW05,CharalambidisHRW13},
can be extended to apply to programs with negation. More specifically, as demonstrated
in~\cite{CharalambidisER14}, every higher-order logic program with negation has a
distinguished extensional model constructed over a logic with an infinite number of
truth values. It is therefore natural to wonder whether the alternative extensional technique
introduced in~\cite{Bezem99,Bezem01}, can also be extended to higher-order logic programs
with negation. It is exactly this question that we answer affirmatively. This brings
us to the following contributions of the present paper:
\begin{itemize}
\item We extend the technique of~\cite{Bezem99,Bezem01} to the class of higher-order
      logic programs with negation. In this way we demonstrate that Bezem's approach
      is more widely applicable than possibly initially anticipated. Our extension
      relies on the {\em infinite-valued semantics}~\cite{RondogiannisW05}, a technique
      that was developed in order to provide a purely model-theoretic semantics
      for negation in classical logic programming.

\item The extensional semantics we propose appears to be simpler compared to~\cite{CharalambidisER14}
      because it relies on the ground instantiation of the higher-order program and does not require
      the rather involved domain-theoretic constructions of~\cite{CharalambidisER14}.
      However, each technique has its own merits and we believe that both will prove to be
      useful tools in the further study of higher-order logic programming.

\item As a case study of the applicability of the new semantics, we define the notions
      of {\em stratification} and {\em local stratification} for higher-order logic programs
      with negation and demonstrate that for such programs the proposed
      semantics never assigns the {\em unknown} truth value. These two new notions generalize the
      corresponding ones from classical (first-order) logic programming. It is worth mentioning that
      such notions have not yet been studied under the semantics of~\cite{CharalambidisER14}.

\item We demonstrate that not all semantic approaches that have been successful for
      classical logic programs with negation lead to extensional semantics when applied
      to higher-order logic programming under the framework of~\cite{Bezem99,Bezem01}.
      In particular we demonstrate that the well-known stable model semantics of
      classical logic programming~\cite{GL88}, if extended according to the technique
      of~\cite{Bezem99,Bezem01} to higher-order logic programs with
      negation, does not in general lead to extensional stable models.
\end{itemize}

The rest of the paper is organized as follows. Section~\ref{intuitive} presents in an intuitive
way the semantics that will be developed in this paper. Section~\ref{infinite_valued} contains
background material on the infinite-valued semantics that will be the basis of our construction.
Section~\ref{syntax_of_language} introduces the syntax and Section~\ref{semantics_of_language} the
semantics of our source language. Section~\ref{extensionality} demonstrates that the proposed semantics is extensional.
In Section~\ref{stratification_section} the notions of {\em stratification} and {\em local stratification}
for higher-order logic programs are introduced. Section~\ref{stable} demonstrates that the stable
model semantics do not in general lead to extensional stable models when applied to higher-order logic
programs. Section~\ref{conclusions} concludes the paper by discussing the connections of the
present work with that of Zolt\'{a}n \'{E}sik, to whom the present special issue is devoted, 
and by providing pointers to future work.

The present paper extends and revises the conference paper~\cite{RS16}. More specifically,
the present paper contains complete proofs of all the claimed results
(which were either missing or sketched in~\cite{RS16}), it contains the new Section ~\ref{stable},
and it has more detailed and polished material in most of the remaining sections.

\section{An Intuitive Overview of the Proposed Approach}\label{intuitive}
In this paper we consider the semantic technique for positive higher-order logic programs proposed in~\cite{Bezem99,Bezem01}, and we extend it in order to apply to programs with negation in clause bodies. Given a positive higher-order logic program, the starting idea behind Bezem's approach is to take its ``ground instantiation'', in which we replace variables with well-typed terms of the Herbrand Universe of the program (i.e., terms that can be created using only
predicate constants, function symbols, and individual constants that appear in the program). For example, consider the higher-order program below
(for the moment we use ad-hoc Prolog-like syntax):
\[
\begin{array}{l}
\mbox{\tt q(a).}\\
\mbox{\tt q(b).}\\
\mbox{\tt p(Q):-Q(a).}\\
\mbox{\tt id(R)(X):-R(X).}
\end{array}
\]
In order to obtain the ground instantiation of this program, we consider each clause
and replace each variable of the clause with a ground term that has the same type
as the variable under consideration (the formal definition of this procedure will be given in
Definition~\ref{ground_instantiation_definition}). In the above example, the variable {\tt X}
represents a data object, so it could be replaced by {\tt a} or {\tt b}, namely the only individual
constant symbols that appear in the program. The variables {\tt P} and {\tt Q}
have the same type, as they both represent a predicate that maps data objects to truth values
(in other words, a unary first-order predicate). Each of them could be replaced by any such
predicate that appears in the program, such as {\tt q}, or a more complex expression that also
maps data objects to truth values, such as {\tt id(q)}. Therefore, the above program defines
a relation {\tt q} that is true of {\tt a} and {\tt b}; a relation {\tt p} that is true of
a unary first-order relation {\tt Q} only if {\tt Q} is true of {\tt a}; and a relation {\tt id}
that is true of a unary first-order relation {\tt R} and a data object {\tt X} only if {\tt R}
is true of {\tt X}. The ground instantiation we obtain is the following infinite program:
\[
\begin{array}{l}
\mbox{\tt q(a).}\\
\mbox{\tt q(b).}\\
\mbox{\tt p(q):-q(a).}\\
\mbox{\tt id(q)(a):-q(a).}\\
\mbox{\tt id(q)(b):-q(b).}\\
\mbox{\tt p(id(q)):-id(q)(a).}\\
\mbox{\tt id(id(q))(a):-id(q)(a).}\\
\mbox{\tt id(id(q))(b):-id(q)(b).}\\
\hspace{2cm} \cdots
\end{array}
\]
One can now treat the new program as an infinite propositional one (i.e., each ground atom can be seen
as a propositional variable). This implies that we can use the standard least fixed-point construction
of classical logic programming (see for example~\cite{lloyd}) in order to compute the set of atoms
that should be taken as ``true''.  In our example, the least fixed-point will
contain atoms such as {\tt q(a)}, {\tt q(b)}, {\tt p(q)}, {\tt id(q)(a)}, {\tt id(q)(b)}, {\tt p(id(q))},
and so on.

The main contribution of Bezem's work was that he established that the least fixed-point
semantics of the ground instantiation of every positive higher-order logic program of the
language considered in~\cite{Bezem99,Bezem01},  is {\em extensional} in a sense that
can be intuitively explained as follows. In the above example {\tt q} and {\tt id(q)} 
can be considered equal since they are both true of exactly the constants {\tt a} and {\tt b}.
Therefore, we would expect that (for example) if {\tt p(q)} is true then {\tt p(id(q))}
is also true, because {\tt q} and {\tt id(q)} should be considered as interchangeable.
This property of ``interchangeability'' is formally defined in~\cite{Bezem99,Bezem01} and it is
demonstrated that it holds in the least fixed-point of the immediate consequence operator
of the ground instantiation of every program.

The key idea behind extending Bezem's semantics in order to apply to higher-order logic programs
with negation, is straightforward to state: given such a program, we first take its ground instantiation.
The resulting program is a (possibly infinite) propositional program with negation,
and therefore we can compute its semantics in any standard way that exists for obtaining
the meaning of such programs. For example, one could use the well-founded semantics~\cite{GelderRS91}
or the stable model semantics~\cite{GL88}, and then proceed to show that the well-founded model
(respectively, each stable model) is extensional in the sense of~\cite{Bezem99,Bezem01} (informally
described above). Instead of using the well-founded or the stable model semantics, we have chosen
to use a relatively recent proposal for assigning meaning to classical logic programs with negation,
namely the infinite-valued semantics~\cite{RondogiannisW05}. As it has been demonstrated
in~\cite{RondogiannisW05}, the infinite-valued semantics is compatible with the well-founded:
if we collapse the infinite-valued model to three truth values, we get the well-founded one.
There are three main reasons for choosing to proceed with the infinite-valued approach
(instead of the well-founded or the stable model approaches):
\begin{itemize}
\item Despite the close connections between the well-founded and the infinite-valued approaches,
      it has recently been demonstrated~\cite{RS17} that the well-founded-based adaptation of the technique 
      of~\cite{Bezem99,Bezem01} is {\em not} extensional.\footnote{The result in~\cite{RS17} was obtained while
      the present paper was under review.} Moreover, as we demonstrate in Section~\ref{stable}, 
      a stable-models-based adaptation of Bezem's technique, does {\em not} in general lead to extensional models.

\item An extension of the infinite-valued approach was used in~\cite{CharalambidisER14} to give the first
      extensional semantics for higher-order logic programs with negation. By developing our present approach
      using the same underlying logic, we facilitate the future comparison between the two approaches.

\item As it was recently demonstrated in~\cite{Esik15,CarayolE16}, the infinite-valued approach satisfies
      all identities of {\em iteration theories}~\cite{BloomE93}, while the well-founded semantics does not.
      Since iteration theories (intuitively) provide an abstract framework for the evaluation of the merits
      of various semantic approaches for languages that involve recursion, the results just mentioned give
      an extra incentive for the further study and use of the infinite-valued approach.
\end{itemize}
We demonstrate that the infinite-valued semantics of the ground instantiation
of every higher-order logic program with negation, is extensional. In this way we extend the results of~\cite{Bezem99,Bezem01}
which applied only to positive programs. The proof of extensionality is quite intricate and is performed
by a tedious induction on the approximations of the minimum infinite-valued model. As an immediate
application of our result, we show how one can define the notions of {\em stratification} and {\em local stratification}
for higher-order logic programs with negation.

\section{The Infinite-valued Semantics}\label{infinite_valued}
In this section we give an overview of the infinite-valued approach of~\cite{RondogiannisW05}.
As in~\cite{RondogiannisW05}, we consider (possibly countably infinite) propositional programs, consisting
of clauses of the form $\mathsf{p}\leftarrow \mathsf{L}_1,\ldots,\mathsf{L}_n$, where each $\mathsf{L}_i$
is either a propositional variable or the negation of a propositional variable; the $\mathsf{L}_i$ will
be called {\em literals}, {\em negative} if they have negation and {\em positive} otherwise. For some
technical reasons that will be explained just after Definition~\ref{ground_instantiation_definition},
we allow a positive literal $\mathsf{L}_i$ to also be one of the constants $\mathsf{true}$ and
$\mathsf{false}$.

The key idea of the infinite-valued approach is that, in order to give a logical semantics to negation-as-failure
and to distinguish it from ordinary negation, one needs to extend the domain of truth values. For example,
consider the program:
\[
\begin{array}{lll}
 {\tt p} & \leftarrow & \\
 {\tt r} & \leftarrow & \pnot {\tt p}\\
 {\tt s} & \leftarrow & \pnot {\tt q}\\
 {\tt t} & \leftarrow & \pnot {\tt t}
\end{array}
\]
According to negation-as-failure, both {\tt p} and {\tt s} receive the value {\em True}.
However, {\tt p} seems ``truer'' than {\tt s} because there
is a rule which says so, whereas {\tt s} is true only because we are never obliged to make
{\tt q} true. In a sense, {\tt s} is true only by default. For this reason, it was proposed
in~\cite{RondogiannisW05} to introduce a ``default'' truth value $T_1$ just below the ``real''
true $T_0$, and (by symmetry) a weaker false value $F_1$ just above (``not as false as'')
the real false $F_0$. Then, negation-as-failure is a combination of ordinary negation with
a weakening. Thus $\pnot F_0 = T_1$ and $\pnot T_0 = F_1$. Since negations can be iterated,
the new truth domain has a  sequence $\ldots,T_3, T_2, T_1$ of weaker and weaker truth values
below $T_0$ but above a neutral value $0$; and a mirror image sequence $F_1, F_2,F_3,\ldots$ above
$F_0$ and below $0$. Since our propositional programs are possibly countably infinite, we need a
$T_\alpha$ and a $F_\alpha$ for every countable ordinal $\alpha$. The intermediate truth value $0$ is
needed for certain atoms that have a ``pathological'' negative dependence on themselves (such as
{\tt t} in the above program). In conclusion, our truth domain $\mathbb{V}$ is shaped as follows:
$$F_0<F_1<\cdots<F_\omega<\cdots<F_\alpha<\cdots<0<\cdots<T_\alpha<\cdots<T_\omega<\cdots<T_1<T_0$$
and the notion of ``Herbrand interpretation of a program'' can be generalized:
\begin{defi}
An (infinite-valued) \emph{interpretation} $I$ of a propositional program $\mathsf{P}$ is a function from the
set of propositional variables that appear in $\mathsf{P}$ to the set $\mathbb{V}$ of truth values.
\end{defi}

For example, an infinite-valued interpretation for the program in the beginning of this section is
$I=\{({\tt p},T_3),({\tt q},F_0),({\tt r},0),({\tt s},F_2),({\tt t},T_0)\}$. As we are going to
see later in this section, the interpretation that captures the meaning of the above program is
$J=\{({\tt p},T_0),({\tt q},F_0),({\tt r},F_1),({\tt s},T_1),({\tt t},0)\}$.

We will use $\emptyset$ to denote the interpretation that assigns the $F_0$ value to all propositional
variables of a program.
If $v \in \mathbb{V}$ is a truth value, we will use $I\parallel v$ to denote the set of variables which are assigned the value $v$ by $I$.
In order to define the notion of ``model'', we need the following definitions:
\begin{defi}\label{interpretation}
Let $I$ be an interpretation of a given propositional program $\mathsf{P}$. For every negative
literal $\pnot \mathsf{p}$ appearing in $\mathsf{P}$ we extend $I$ as follows:
\[
             I(\pnot \mathsf{p}) = \left\{
                             \begin{array}{ll}
                             T_{\alpha + 1}, & \mbox{ if $I(\mathsf{p}) = F_\alpha$}\\
                             F_{\alpha + 1}, & \mbox{ if $I(\mathsf{p}) = T_\alpha$}\\
                             0,              & \mbox{ if $I(\mathsf{p}) = 0$}
                             \end{array}
                      \right.
\]
Moreover, $I(\mathsf{true}) = T_0$ and
$I(\mathsf{false}) = F_0$. Finally, for every conjunction of literals $\mathsf{L}_1,\ldots,\mathsf{L}_n$ appearing as
the body of a clause in $\mathsf{P}$, we extend $I$ by $I(\mathsf{L}_1,\ldots,\mathsf{L}_n) = min\{I(\mathsf{L}_1),\ldots,I(\mathsf{L}_n)\}$.
\end{defi}
\begin{defi}
Let $\mathsf{P}$ be a propositional program and $I$ an interpretation of $\mathsf{P}$. Then, $I$ {\em satisfies}
a clause $\mathsf{p} \leftarrow \mathsf{L}_1,\ldots,\mathsf{L}_n$ of $\mathsf{P}$ if $I(\mathsf{p}) \geq I(\mathsf{L}_1,\ldots,\mathsf{L}_n)$. Moreover, $I$ is a {\em model} of $\mathsf{P}$ if $I$ satisfies all
clauses of $\mathsf{P}$.
\end{defi}
As it is demonstrated in~\cite{RondogiannisW05}, every program has a {\em minimum}
infinite-valued model under an ordering relation $\aleq$, which compares interpretations in
a stage-by-stage manner. To formally state this result, the following definitions are necessary:
\begin{defi}
The \emph{order} of a truth value is defined as follows: $order(T_\alpha)=\alpha$, $order(F_\alpha)=\alpha$ and $order(0)=+\infty$.
\end{defi}
\begin{defi}
Let $I$ and $J$ be interpretations of a given propositional program $\mathsf{P}$ and $\alpha$ be a countable
ordinal. We write $I=_{\alpha} J$, if for all $\beta \leq \alpha$,
$I\parallel T_{\beta} = J\parallel T_{\beta}$ and $I\parallel F_{\beta} = J\parallel F_{\beta}$. We write
$I\sqsubseteq_{\alpha} J$, if for all $\beta < \alpha$, $I=_{\beta} J$ and, moreover,  $I\parallel T_{\alpha} \subseteq J\parallel T_{\alpha}$ and $I\parallel F_{\alpha} \supseteq J\parallel F_{\alpha}$.
We write $I \sqsubset_{\alpha} J$, if $I \sqsubseteq_{\alpha} J$ but $I =_{\alpha} J$ does not hold.
\end{defi}
\begin{defi}
Let $I$ and $J$ be interpretations of a given propositional program $\mathsf{P}$. We write $I\ale J$,
if there exists a countable ordinal $\alpha$  such that
$I \sqsubset_{\alpha} J$. We write $I\aleq J$ if either
$I = J$ or $I\ale J$.
\end{defi}
It is easy to see~\cite{RondogiannisW05} that $\aleq$  is a partial order, $\aleq[\alpha]$ is a preorder,
and $=_\alpha$ is an equivalence relation.
As in the case of positive programs, the minimum model of a program $\mathsf{P}$ coincides with the least fixed-point of an operator $T_\mathsf{P}$. 
This operator is defined through the notion of the ``least upper bound'' of a set of truth values. 
\begin{defi}
Let $S$ be a set with a partial order $\leq$ and let $A\subseteq S$. We say that $u\in S$ is an {\em upper bound} of $A$, 
if for every $v \in A$ we have $v\leq u$. Moreover, $u$ is called the {\em least upper bound} of $A$, if for every upper bound 
$u'$ of $A$ we have $u\leq u'$.
\end{defi}
If the least upper bound of $A$ exists then it is unique and we denote it by $lub(A)$. In~\cite{RondogiannisW05} it is shown that every subset of $\mathbb{V}$ has a least upper bound and the operator $T_\mathsf{P}$ can then be defined as below:
\begin{defi}
Let $\mathsf{P}$ be a propositional program and let $I$ be an interpretation of $\mathsf{P}$.
The {\em immediate consequence operator} $T_{\mathsf{P}}$ of $\mathsf{P}$ is defined as follows:
    $$T_\mathsf{P}(I)(\mathsf{p}) = lub(\{I(\mathsf{L}_1,\ldots,\mathsf{L}_n) \mid \mathsf{p}\leftarrow \mathsf{L}_1,\ldots,\mathsf{L}_n \in \mathsf{P}\})$$
\end{defi}

The least fixed-point $M_\mathsf{P}$ of $T_\mathsf{P}$ is constructed as follows. We start with $\emptyset$,
namely the interpretation that assigns to every propositional variable of $\mathsf{P}$ the value $F_0$. We iterate $T_\mathsf{P}$ on $\emptyset$
until the set of variables having a $F_0$ value and the set of variables having a $T_0$ value, stabilize. Then we reset the values of
all remaining variables to $F_1$. The procedure is repeated until the $F_1$ and $T_1$ values stabilize, and we reset the remaining variables
to $F_2$, and so on. It is shown in~\cite{RondogiannisW05} that there exists a countable ordinal $\delta$ for which this process will not produce any new variables having $F_\delta$ or $T_\delta$ values. At this point we reset all remaining variables to 0. The following definitions formalize this process.
\begin{defi}
Let $\mathsf{P}$ be a propositional program and let $I$ be an interpretation of $\mathsf{P}$. For each countable ordinal $\alpha$, we define the
interpretation $T_{\mathsf{P},\alpha}^{\omega}(I)$ as follows:
      \[
             T_{\mathsf{P},\alpha}^{\omega}(I)(\mathsf{p}) = \left\{
                    \begin{array}{ll}
                    I(\mathsf{p}), & \mbox{ if $order(I(\mathsf{p})) < \alpha$}\\
                    T_{\alpha},     & \mbox{ if $\mathsf{p}\in \bigcup_{n < \omega}(T_\mathsf{P}^n(I) \parallel T_{\alpha})$}\\
                    F_{\alpha},     & \mbox{ if $\mathsf{p}\in \bigcap_{n < \omega}(T_\mathsf{P}^n(I) \parallel F_{\alpha})$}\\
                    F_{\alpha+1},   & \mbox{ otherwise}
                    \end{array}
                \right.
      \]
\end{defi}
\begin{defi}\label{approx}
Let $\mathsf{P}$ be a propositional program. For each countable ordinal $\alpha$, we define
$M_\alpha = T^{\omega}_{\mathsf{P},\alpha}(\baseap{\alpha})$ where $ \baseap{0}=\emptyset$,
$\baseap{\alpha} = M_{\alpha-1}$ if $\alpha$ is a successor ordinal, and
\[\begin{array}{ccll}
   \baseap{\alpha}(\mathsf{p}) & = & \left\{
                    \begin{array}{ll}
                      M_{\beta}(\mathsf{p}),& \mbox{ if $order(M_{\beta}(\mathsf{p}))=\beta$ for some ${\beta < \alpha}$}\\
                      F_{\alpha},& \mbox{ otherwise}
                    \end{array}
                    \right.
\end{array}
\]
if $\alpha$ is a limit ordinal. The $M_0,M_1,\ldots,M_{\alpha},\ldots$ are called the {\em approximations} to the minimum model of $\mathsf{P}$.
\end{defi}
In~\cite{RondogiannisW05} it is shown that the above sequence of approximations is well-defined.
We will make use of the following lemma from~\cite{RondogiannisW05}:
\begin{lem}
\label{lm:TPn-and-TPomega}
Let $\mathsf{P}$ be a propositional program and let $\alpha$ be a countable ordinal. For all $n<\omega$, $T^n_\mathsf{P}(\baseap{\alpha})\aleq[\alpha]M_\alpha$.
\end{lem}
The following lemma from~\cite{RondogiannisW05} states that there exists a certain countable ordinal, after which
new approximations do not introduce new truth values:
\begin{lem}\label{depthlemma}
Let $\mathsf{P}$ be a propositional program. Then, there exists a countable ordinal $\delta$, called the {\em depth} of $\mathsf{P}$,
such that:
\begin{enumerate}
\item for all countable ordinals $\gamma \geq \delta$,
$M_{\gamma} \parallel T_{\gamma} = \emptyset$ and $M_{\gamma} \parallel F_{\gamma} = \emptyset$;
\item for all $\beta < \delta$, $M_{\beta} \parallel T_{\beta} \neq \emptyset$ or $M_{\beta} \parallel F_{\beta} \neq \emptyset$.
\end{enumerate}
\end{lem}
Given a propositional program $\mathsf{P}$ that has depth $\delta$, we define the
following interpretation $M_\mathsf{P}$:
\[
       M_\mathsf{P}(\mathsf{p}) = \left\{
                    \begin{array}{ll}
                      M_{\delta}(\mathsf{p}),& \mbox{if $order(M_{\delta}(\mathsf{p}))<\delta$}\\
                      0, & \mbox{otherwise}
                    \end{array}
                  \right.
\]
The following two theorems from~\cite{RondogiannisW05}, establish interesting properties of $M_\mathsf{P}$:
\begin{thm}\label{minimum}
The infinite-valued interpretation $M_\mathsf{P}$ is a model of $\mathsf{P}$. Moreover, it is
the least (with respect to $\aleq$) among all the infinite-valued models of $\mathsf{P}$.
\end{thm}
\begin{thm}\label{collapses}
The interpretation $N_\mathsf{P}$ obtained by collapsing all true values of $M_\mathsf{P}$
to {\em True} and all false values to {\em False}, coincides with the well-founded model of $\mathsf{P}$.
\end{thm}
The next lemma states a fact already implied earlier, namely that new approximations do not affect the sets of variables stabilized by the preceding ones.
\begin{lem}
\label{lm:approx_alpha_eq}
Let $\mathsf{P}$ be a propositional program and let $\alpha$ be a countable ordinal. For all countable ordinals $\beta>\alpha$, $M_\alpha=_\alpha M_\beta$. Moreover, $M_\alpha=_\alpha M_\mathsf{P}$.
\end{lem}

\section{The Syntax of \texorpdfstring{$\mathcal{H}$}{H}}\label{syntax_of_language}
In this section we define the syntax of the language $\mathcal{H}$ that we use throughout the paper.
$\mathcal{H}$ is based on a simple type system with two base types: $o$, the boolean domain, and $\iota$, the domain of data objects.
The composite types are partitioned into three classes: functional (assigned to function symbols), predicate (assigned to predicate symbols)
and argument (assigned to parameters of predicates).
\begin{defi}
A type can either be \emph{functional}, \emph{predicate}, or \emph{argument}, denoted by $\sigma$, $\pi$
and $\rho$ respectively and defined as:
\begin{align*}
\sigma & := \iota \mid (\iota \rightarrow \sigma) \\
\pi & := o \mid (\rho \rightarrow \pi) \\
\rho & := \iota \mid \pi
\end{align*}
\end{defi}

We will use $\tau$ to denote an arbitrary type (either functional, predicate, or argument). As usual, the binary operator $\rightarrow$ is right-associative. A functional type that is different than $\iota$ will often be written in the form $\iota^n \rightarrow \iota$, $n\geq 1$. Moreover, it can be easily seen that every predicate type $\pi$ can be written in the form $\rho_1 \rightarrow \cdots \rightarrow \rho_n \rightarrow o$, $n\geq 0$ (for $n=0$ we assume that $\pi=o$).  We proceed by defining the syntax of $\mathcal{H}$:
\begin{defi}
The \emph{alphabet} of $\mathcal{H}$ consists of the following:
\begin{enumerate}
  \item Predicate variables of every predicate type $\pi$ (denoted by capital letters such as
      $\mathsf{Q,R,}$ $\mathsf{S,\ldots}$).
  \item Individual variables of type $\iota$ (denoted by capital letters such as
      $\mathsf{X,Y,Z,\ldots}$).
  \item Predicate constants of every predicate type $\pi$ (denoted by lowercase letters such as $\mathsf{p,q,r,\ldots}$).
  \item Individual constants of type $\iota$ (denoted by lowercase
      letters such as $\mathsf{a,b,c,\ldots}$).
  \item Function symbols of every functional type $\sigma \neq \iota$ (denoted by lowercase letters such as $\mathsf{f,g,h,\ldots}$).
  \item The inverse implication constant $\leftarrow$, the negation constant $\pnot$, the comma, the left and right parentheses, and the equality constant $\approx$ for comparing terms of type $\iota$.
\end{enumerate}
\end{defi}
Arbitrary variables will be usually denoted by $\mathsf{V}$ and its subscripted versions.

\begin{defi}
The set of {\em terms} of $\mathcal{H}$ is defined as follows:
\begin{itemize}
  \item Every predicate variable (respectively, predicate constant) of type $\pi$ is a
        term of type $\pi$; every individual variable (respectively, individual constant)
        of type $\iota$ is a term of type $\iota$;
  \item if $\mathsf{f}$ is an $n$-ary function symbol and $\mathsf{E}_1, \ldots, \mathsf{E}_n$
        are terms of type $\iota$ then $(\mathsf{f}\ \mathsf{E}_1\cdots\mathsf{E}_n)$ is
        a term of type $\iota$;
  \item if $\mathsf{E}_1$ is a term of type $\rho \rightarrow \pi$ and
        $\mathsf{E}_2$ a term of type $\rho$ then $(\mathsf{E}_1\ \mathsf{E}_2)$ is a term of type $\pi$.
\end{itemize}
\end{defi}
\begin{defi}
\label{def:expressions}
The set of {\em expressions} of $\mathcal{H}$ is defined as follows:
\begin{itemize}
\item A term of type $\rho$ is an expression of type $\rho$;
\item if $\mathsf{E}$ is a term of type $o$ then $(\pnot \mathsf{E})$ is an expression of type $o$;
\item if $\mathsf{E}_1$ and $\mathsf{E}_2$ are terms of type $\iota$, then $(\mathsf{E}_1\approx \mathsf{E}_2)$ is an expression of type $o$.
\end{itemize}
\end{defi}
We write $vars(\mathsf{E})$ to denote the set of all the variables in $\mathsf{E}$.
Expressions (respectively, terms) that have no variables will often be referred to as {\em ground expressions} (respectively, {\em ground terms}). We will omit parentheses
when no confusion arises. To denote that an expression $\mathsf{E}$ has type $\rho$ we will often write $\mathsf{E}:\rho$.
Terms of type $o$ will often be referred to as {\em atoms}. Expressions of type $o$ that do not contain negation, i.e. atoms and expressions of the form $(\mathsf{E}_1\approx \mathsf{E}_2)$, will be called {\em positive literals}, while expressions of the form $(\pnot \mathsf{E})$ will be called {\em negative literals}. A {\em literal} is either a positive literal or a negative literal.

\begin{defi}\label{programs}
A {\em clause} of $\mathcal{H}$ is a formula
$\mathsf{p}\ \mathsf{V}_1 \cdots \mathsf{V}_n \leftarrow \mathsf{L}_1, \ldots, \mathsf{L}_m$,
where $\mathsf{p}$ is a predicate constant of type $\rho_1 \rightarrow \cdots \rightarrow\rho_n \rightarrow o$, $\mathsf{V}_1,\ldots,\mathsf{V}_n$ are distinct variables of types $\rho_1,\ldots,\rho_n$ respectively and $\mathsf{L}_1,\ldots,\mathsf{L}_m$ are literals. The term $\mathsf{p}\ \mathsf{V}_1 \cdots \mathsf{V}_n$ is called the {\em head} of the clause, the variables $\mathsf{V}_1, \ldots, \mathsf{V}_n$ are the {\em formal parameters} of the
clause
and the conjunction $\mathsf{L}_1,\ldots, \mathsf{L}_m$ is its {\em body}.
A {\em program} $\mathsf{P}$ of $\mathcal{H}$ is a finite set of clauses.
\end{defi}
\begin{exa}
The program below defines the {\tt subset} relation over unary predicates:
\[
\begin{array}{l}
\mbox{\tt subset S1 S2 :- $\pnot$(nonsubset S1 S2).}\\
\mbox{\tt nonsubset S1 S2 :- S1 X, $\pnot$(S2 X).}
\end{array}
\]
The variables {\tt S1} and {\tt S2} are both predicate variables of type $\iota\rightarrow o$, while {\tt X} is an individual variable (i.e., it is of type $\iota$). Given predicates {\tt p} and {\tt q} of type $\iota\rightarrow o$, {\tt subset p q} is true if {\tt p} is a subset of {\tt q}.\qed
\end{exa}
In the following, we will often talk about the ``ground instantiation of a program''. This
notion is formally defined below.
\begin{defi}\label{def:ground_sub}
A {\em substitution} $\theta$ is a finite set of the form $\{ \mathsf{V}_1/\mathsf{E}_1, \ldots, \mathsf{V}_n/\mathsf{E}_n\}$
where the $\mathsf{V}_i$'s are different variables and each $\mathsf{E}_i$
is a term having the same type as $\mathsf{V}_i$. We write
$dom(\theta)$ to denote the domain $\{ \mathsf{V}_1, \ldots, \mathsf{V}_n\}$ of $\theta$. If
all the expressions $\mathsf{E}_1, \ldots, \mathsf{E}_n$ are ground terms, $\theta$ is called a {\em ground substitution}.
\end{defi}
We can now define the application of a substitution to an expression.
\begin{defi}
\label{def:ground_inst}
Let $\theta$ be a substitution and $\mathsf{E}$ be an expression. Then, $\mathsf{E}\theta$
is an expression obtained from $\mathsf{E}$ as follows:
\begin{itemize}
  \item $\mathsf{E}\theta = \mathsf{E}$ if $\mathsf{E}$ is a predicate constant or individual constant;
  \item $\mathsf{V}\theta = \theta(\mathsf{V})$ if $\mathsf{V} \in dom(\theta)$; otherwise, $\mathsf{V}\theta = \mathsf{V}$;
  \item $(\mathsf{f}\ \mathsf{E}_1\cdots\mathsf{E}_n)\theta = (\mathsf{f}\ \mathsf{E}_1\theta\cdots\mathsf{E}_n\theta)$;
  \item $(\mathsf{E}_1\ \mathsf{E}_2)\theta = (\mathsf{E}_1\theta\ \mathsf{E}_2\theta)$;
  \item $(\pnot \mathsf{E})\theta= (\pnot \mathsf{E}\theta)$;
  \item $(\mathsf{E}_1\approx \mathsf{E}_2)\theta = (\mathsf{E}_1\theta\approx \mathsf{E}_2\theta)$.
\end{itemize}
If $\theta$ is a ground substitution such that $vars(\mathsf{E}) \subseteq dom(\theta)$, then
the ground expression $\mathsf{E}\theta$ is called a {\em ground instance} of $\mathsf{E}$.
\end{defi}
\begin{defi}\label{ground_instantiation_definition}
Let $\mathsf{P}$ be a program. A {\em ground instance of a clause}
$\mathsf{p}\ \mathsf{V}_1 \cdots \mathsf{V}_n \leftarrow \mathsf{L}_1,\ldots,\mathsf{L}_m$
of $\mathsf{P}$ is a formula  $(\mathsf{p}\ \mathsf{V}_1 \cdots \mathsf{V}_n)\theta \leftarrow \mathsf{L}_1\theta,\ldots,\mathsf{L}_m\theta$, where $\theta$ is a ground substitution whose domain
is the set of all variables that appear in the clause, such that for every $\mathsf{V} \in dom(\theta)$ with $\mathsf{V}:\rho$,
$\theta(\mathsf{V})$ is a ground expression of type $\rho$ that has been formed with predicate constants, function symbols,
and individual constants that appear in $\mathsf{P}$. The {\em ground instantiation of a program} $\mathsf{P}$,
denoted by $\mathsf{Gr(P)}$, is the (possibly infinite) set that contains all the ground instances of the
clauses of $\mathsf{P}$.
\end{defi}
In the rest of the paper, we will be extensively using the notion of ground instantiation.
Notice that in the body of a clause of the ground instantiation, there may exist ground
expressions of the form $(\mathsf{E}_1\approx \mathsf{E}_2)$. In the case where the two
expressions $\mathsf{E}_1$ and $\mathsf{E}_2$ are syntactically identical, the expression
$(\mathsf{E}_1\approx \mathsf{E}_2)$ will be treated as the constant $\mathsf{true}$, and
otherwise as the constant $\mathsf{false}$.

\section{The Semantics of \texorpdfstring{$\mathcal{H}$}{H}}\label{semantics_of_language}
In this section we develop the semantics of $\mathcal{H}$. Our developments generalize the semantics of~\cite{Bezem99,Bezem01}
for positive higher-order logic programs, to programs with negation. Notice that the semantics of~\cite{Bezem99,Bezem01} is
based on classical two-valued logic, while ours on the infinite-valued logic of Section~\ref{infinite_valued}.

In order to interpret the programs of  $\mathcal{H}$, we need to specify the semantic domains in which the expressions of each type $\tau$ are assigned their meanings. We adopt the approach of~\cite{Bezem99,Bezem01}. More specifically, the following definition implies that the expressions of predicate types should be understood as representing functions. We use $[S_1 \rightarrow S_2]$ to denote the set of (possibly partial) functions from a set $S_1$ to a set $S_2$. The possibility to have a partial function arises due to a technicality
which is explained in the remark just above Definition~\ref{HInterp}.
\begin{defi}
\label{def:type_structure}
A \emph{functional type structure} $\mathcal{S}$ for $\mathcal{H}$ consists of two non-empty sets $D$ and $A$ together with an assignment
$\mo{\tau}$ to each type $\tau$ of $\mathcal{H}$, so that the following are satisfied:

\begin{itemize}
\item $\mo{\iota} = D$;
\item $\mo{\iota^n \rightarrow \iota} = D^n \rightarrow D$;
\item $\mo{o} = A$;
\item $\mo{\rho \rightarrow \pi} \subseteq [\mo{\rho} \rightarrow \mo{\pi}]$.
\end{itemize}
\end{defi}

Given a functional type structure $\mathcal{S}$, any function $\mathit{val}:\mo{o}\rightarrow \mathbb{V}$ will be called an  \emph{infinite-valued valuation function} (or simply \emph{valuation function}) for $\mathcal{S}$. 

It is customary in the study of the semantics of logic programming languages to restrict attention to
{\em Herbrand interpretations}. Given a program $\mathsf{P}$, a Hebrand interpretation is one that
has as its underlying universe the so-called {\em Herbrand universe} of $\mathsf{P}$:
\begin{defi}\label{HUniverse}
For a program $\mathsf{P}$, we define the \emph{Herbrand universe} for every argument type $\rho$, denoted by
$U_{\mathsf{P},\rho}$, to be the set of all ground terms of type $\rho$ that can be formed out of the individual constants, function symbols, and predicate constants in the program. Moreover, we define $U^+_{\mathsf{P},o}$ to be the set of all ground expressions of type $o$, that can be formed out of the above symbols, i.e. the set $U^+_{\mathsf{P},o}=U_{\mathsf{P},o} \cup \{(\mathsf{E}_1\approx \mathsf{E}_2)\mid\mathsf{E}_1, \mathsf{E}_2\in U_{\mathsf{P},\iota}\} \cup \{(\pnot \mathsf{E})\mid\mathsf{E}\in U_{\mathsf{P},o}\}$.
\end{defi}

Following~\cite{Bezem99,Bezem01}, we take $D$ and $A$ in Definition \ref{def:type_structure} to be equal to $U_{\mathsf{P},\iota}$ and $U^+_{\mathsf{P},o}$ respectively. Then, each element of $U_{\mathsf{P}, \rho \rightarrow \pi}$ can itself be perceived as a function mapping elements of $\mo{\rho}$ to elements of $\mo{\pi}$, through syntactic application mapping. That is, $\mathsf{E}\in U_{\mathsf{P},\rho \rightarrow \pi}$ can be viewed as the function mapping each term $\mathsf{E}' \in  U_{\mathsf{P},\rho}$ to the term $\mathsf{E} \, \mathsf{E}' \in  U_{\mathsf{P},\pi}$. Similarly, every $n$-ary function symbol $\mathsf{f}$ appearing in $\mathsf{P}$ can be viewed as the function mapping each element $(\mathsf{E}_1, \ldots,\mathsf{E}_n) \in U_{\mathsf{P},\iota}^n$ to the term $(\mathsf{f} \, \mathsf{E}_1\,\cdots\,\mathsf{E}_n) \in U_{\mathsf{P},\iota}$.

\vspace{0.2cm}
\noindent
{\bf Remark:} There is a small technicality here which we need to clarify. In the case where $\rho=o$,
$\mathsf{E}\in U_{\mathsf{P},o \rightarrow \pi}$ is a partial function because it maps elements of
$U_{\mathsf{P},o}$ (and not of $U^+_{\mathsf{P},o}$) to elements of $U_{\mathsf{P},\pi}$; this is due
to the fact that our syntax does not allow an expression of type $o \rightarrow \pi$ to take as argument
an expression of the form $(\mathsf{E}_1\approx \mathsf{E}_2)$ nor of the form $(\pnot \mathsf{E})$.
In all other cases (i.e., when $\rho\neq o$), $\mathsf{E}$ represents a total function.

\begin{defi}\label{HInterp}
A \emph{Herbrand interpretation} $I$ of a program $\mathsf{P}$ consists of
\begin{enumerate}
\item a functional type structure $\mathcal{S}_I$, such that $D= U_{\mathsf{P},\iota}$, $A=U^+_{\mathsf{P},o}$ and $\mo{\rho \rightarrow \pi} =U_{\mathsf{P},\rho \rightarrow \pi}$ for every predicate type $\rho \rightarrow \pi$;
\item an assignment to each individual constant $\mathsf{c}$ in $\mathsf{P}$, of the element $I(\mathsf{c}) = \mathsf{c}$; to each predicate constant $\mathsf{p}$ in $\mathsf{P}$, of the element $I(\mathsf{p}) =\mathsf{p}$; to each function symbol $\mathsf{f}$ in $\mathsf{P}$, of the element $I(\mathsf{f}) = \mathsf{f}$;
\item a valuation function $\vwrt{\cdot}{I}$ for $\mathcal{S}_I$,  assigning to each element of $U^+_{\mathsf{P},o}$ an element in  $\mathbb{V}$, while satisfying the following:
	\begin{itemize}
	\item for all $\mathsf{E}_1,\mathsf{E}_2\in U_{\mathsf{P},\iota}$, $\vwrt{(\mathsf{E}_1\approx\mathsf{E}_2)}{I}= \begin{cases}
F_{0}, & \mbox{if } \mathsf{E}_1\neq \mathsf{E}_2\\
T_{0}, & \mbox{if } \mathsf{E}_1= \mathsf{E}_2
\end{cases}$;
	\item for all $\mathsf{E}\in U_{\mathsf{P},o}$, $\vwrt{(\mathsf{\pnot E)}}{I} = \begin{cases}
T_{\alpha+1}, & \mbox{if } \vwrt{\mathsf{E}}{I}=F_\alpha\\
F_{\alpha+1}, & \mbox{if } \vwrt{\mathsf{E}}{I}=T_\alpha\\
\unk, & \mbox{if } \vwrt{\mathsf{E}}{I}=\unk
\end{cases}$.
	\end{itemize}
\end{enumerate}
\end{defi}
We call $\vwrt{\cdot}{I}$ the {\em valuation function of $I$} and omit the reference to $\mathcal{S}_I$, since the latter is common to all Herbrand interpretations of a program. In fact, individual Herbrand interpretations are only set apart by their valuation functions.

\begin{defi}\label{def:HState}
  A \emph{Herbrand state} (or simply \emph{state}) $s$ of a program $\mathsf{P}$ is a function that assigns to each variable $\mathsf{V}$ of type $\rho$ an element of $U_{\mathsf{P},\rho}$.
\end{defi}
Given a Herbrand interpretation $I$ and state $s$, we can define the semantics of expressions with respect to $I$ and $s$.
\begin{defi}
Let $\mathsf{P}$ be a program. Also, let $I$ be a Herbrand interpretation and $s$ a
Herbrand state of $\mathsf{P}$. Then the semantics of expressions with respect to $I$ and $s$ is defined as follows:
\begin{itemize}
\item $\mwrs{\mathsf{c}}{I}{s} = I(\mathsf{c})=\mathsf{c}$, for every individual constant $\mathsf{c}$;
\item $\mwrs{\mathsf{p}}{I}{s} =I(\mathsf{p})=\mathsf{p}$, for every predicate constant $\mathsf{p}$;
\item $\mwrs{\mathsf{V}}{I}{s} = s(\mathsf{V})$, for every variable $\mathsf{V}$;
\item $\mwrs{(\mathsf{f}\, \mathsf{E}_1\, \cdots\,\mathsf{E}_n)}{I}{s}
       =(I(\mathsf{f})\, \mwrs{\mathsf{E_1}}{I}{s}\, \cdots\,\mwrs{\mathsf{E_n}}{I}{s})
     = (\mathsf{f}\, \mwrs{\mathsf{E}_1}{I}{s}\, \cdots\,\mwrs{\mathsf{E}_n}{I}{s})$, for every $n$-ary function symbol $\mathsf{f}$;
\item $\mwrs{(\mathsf{E}_1\,\mathsf{E}_2)}{I}{s}= (\mwrs{\mathsf{E}_1}{I}{s}\,\mwrs{\mathsf{E}_2}{I}{s})$;
\item $\mwrs{(\mathsf{E}_1\approx \mathsf{E}_2)}{I}{s}= (\mwrs{\mathsf{E}_1}{I}{s}\approx\mwrs{\mathsf{E}_2}{I}{s})$;
\item $\mwrs{\mathsf{(\pnot E)}}{I}{s}= (\pnot \mwrs{\mathsf{E}}{I}{s}).$
\end{itemize}
\end{defi}

Since we are dealing with Herbrand interpretations, it is easy to see that for every Herbrand state $s$ and ground expression $\mathsf{E}$, we have $\mwrs{\mathsf{E}}{I}{s} = \mathsf{E}$.
Therefore, if $\mathsf{E}$ is a ground literal, we can write $\vwrt{\mathsf{E}}{I}$ instead of $\vwrt{\mwrs{\mathsf{E}}{I}{s}}{I}$.
Stretching this abuse of notation a little further, we can extend a valuation function to assign truth values to ground conjunctions of literals:

\begin{defi}
Let $\mathsf{P}$ be a program and $I$ be a Herbrand interpretation of $\mathsf{P}$. We define $\vwrt{\mathsf{L}_1, \ldots, \mathsf{L}_n}{I} = min\{\vwrt{\mathsf{L}_1}{I}, \ldots, \vwrt{\mathsf{L}_n}{I}\}$
for all $\mathsf{L}_1, \ldots, \mathsf{L}_n\in U^+_{\mathsf{P},o}$.
\end{defi}
Based on the above definition, we can define the concept of Herbrand models for our higher-order programs in the same way as in classical logic programming.
\begin{defi}
Let $\mathsf{P}$ be a program and $I$ be a Herbrand interpretation of $\mathsf{P}$. We say $I$ is a \emph{model} of $\mathsf{P}$ if $\vwrt{\mwrs{\mathsf{A}}{I}{s}}{I}\geq\vwrt{\mwrs{\mathsf{L}_1}{I}{s},\ldots, \mwrs{\mathsf{L}_m}{I}{s}}{I}$ holds for every clause $\mathsf{A} \leftarrow \mathsf{L}_1,\ldots, \mathsf{L}_m$ and every Herbrand state $s$ of $\mathsf{P}$.
\end{defi}
Bezem's semantics is based on the observation that, given a positive higher-order program, we can use the minimum model semantics
of its ground instantiation as a (two-valued) valuation function defining a Herbrand interpretation for the initial program itself. We use
the same idea for $\mathcal{H}$ programs; the only difference is that we employ the infinite-valued model of the ground instantiation of the
program as the valuation function.

\hide{
Bezem's semantics is based on the simple observation that the ground instantiation of a higher-order program could be treated as a propositional program. This allows us to calculate the infinite-valued model of this ground ``program'' and use the truth values it assigns to ground atoms, in order to assign meaning to the higher-order program. In essence, we can use the infinite-valued model, as well as any other Herbrand interpretation of the ground instantiation of the program, as a valuation function defining an interpretation for the higher-order program itself. We exploit this idea in the following definition and on other instances throughout the rest of the paper.
}
\begin{defi}
Let $\mathsf{P}$ be a program. Also, let $\mathsf{Gr(P)}$ be the ground instantiation of $\mathsf{P}$ and let $\wgp$ be the infinite-valued model of $\mathsf{Gr(P)}$. We define $\wfmp$ to be the Herbrand interpretation of $\mathsf{P}$ such that $\vwrt{\mathsf{A}}{\wfmp} = \wgp(\mathsf{A})$ for every $\mathsf{A}\in U_{\mathsf{P},o}$.
\end{defi}

We adopt the notation $I\parallel v$ from Section \ref{infinite_valued}, to signify the set of atoms which are assigned a certain truth value $v \in \mathbb{V}$ by a Herbrand interpretation $I$; that is,  $I\parallel v=\{\mathsf{A} \mid \mathsf{A}\in U_{\mathsf{P},o} \mbox{ and } \vwrt{\mathsf{A}}{I}=v\}$. Then the relations $\aleq[\alpha]$, $\ale[\alpha]$, $=_\alpha$, $\aleq$ and $\ale$ on Herbrand interpretations of a higher-order program can be defined in exactly the same manner as in Section~\ref{infinite_valued}.

The next theorem verifies that our semantics is well-defined, in the sense that the interpretation $\wfmp$, which we chose as the meaning of a program $\mathsf{P}$, is indeed a model of $\mathsf{P}$. In fact it is the minimum, with respect to $\aleq$, model of $\mathsf{P}$.
\hide{

The following lemma will be useful for this purpose.
\begin{lem}
\label{lm:substitution_lemma}
Let $\mathsf{P}$ be a program, $I$ be an interpretation of $\mathsf{P}$ and $\mathsf{E}$ be an expression of some type $\rho$. Also, let $s$ be a Herbrand state of $\mathsf{P}$ and $\theta$ be a ground substitution with $vars(\mathsf{E}) \subseteq dom(\theta)$ and such that $s(\mathsf{V}) = \mwrs{\theta(\mathsf{V})}{I}{s'}$ for every variable $\mathsf{V}$ in $vars(\mathsf{E})$ and every state $s'$. Then $\mwrs{\mathsf{E}\theta}{I}{s'}= \mwrs{\mathsf{E}}{I}{s}$ for all $s'$.
\end{lem}
\begin{proof}
Trivial, by induction on the structure of $\mathsf{E}$.
\end{proof}

}
\begin{thm}
\label{thr:min_mod}
$\wfmp$ is the minimum (with respect to $\aleq$) Herbrand model of $\mathsf{P}$.
\end{thm}
\begin{proof}
Let $\mathsf{Gr(P)}$ be the ground instantiation of $\mathsf{P}$ and $\wgp$ be the infinite-valued model of $\mathsf{Gr(P)}$.
Recall (Definition~\ref{def:HState}) that a Herbrand state $s$ of $\mathsf{P}$ assigns to each variable $\mathsf{V}$ of type $\rho$ an element of $U_{\mathsf{P},\rho}$ and that this (semantic) element is in fact a ground expression of the same type. Moreover, recall (Definition~\ref{def:ground_sub}) that a ground substitution, which includes $\mathsf{V}$ in its domain, also assigns a ground expression of type $\rho$ to $\mathsf{V}$; and that such an expression coincides with its own meaning, under any Herbrand interpretation and any state. Therefore,
for every Herbrand state $s$ of $\mathsf{P}$ there exists a ground substitution $\theta$ such that $s(\mathsf{V}) = \mwrs{\theta(\mathsf{V})}{\wfmp}{s'}$ for all states $s'$ and variables $\mathsf{V}$ in $\mathsf{P}$. Also, for every clause $\mathsf{A} \leftarrow \mathsf{L}_1, \ldots, \mathsf{L}_m$ in $\mathsf{P}$ there exists a respective ground instance $\mathsf{A}\theta \leftarrow \mathsf{L}_1\theta, \ldots, \mathsf{L}_m\theta$ in $\mathsf{Gr(P)}$. As $\wgp$ is a model of $\mathsf{Gr(P)}$, $\wgp(\mathsf{A}\theta)\geq min\{\wgp(\mathsf{L}_1\theta), \ldots, \wgp(\mathsf{L}_m\theta)\}$. By definition, $\vwrt{\mathsf{A}\theta}{\wfmp}= \wgp(\mathsf{A}\theta)$ and $\vwrt{\mathsf{L}_i\theta}{\wfmp}= \wgp(\mathsf{L}_i\theta)$ for all $i\leq m$. Moreover, it is easy to see (by a trivial induction on the structure of the expression) that $\mathsf{A}\theta= \mwrs{\mathsf{A}}{\wfmp}{s}$, which implies that $\vwrt{\mathsf{A}\theta}{\wfmp}= \vwrt{\mwrs{\mathsf{A}}{\wfmp}{s}}{\wfmp}$. Similarly, $\vwrt{\mathsf{L}_i\theta}{\wfmp}= \vwrt{\mwrs{\mathsf{L}_i}{\wfmp}{s}}{\wfmp}$, for all $i\leq m$. Then $\vwrt{\mwrs{\mathsf{A}}{\wfmp}{s}}{\wfmp}\geq min\{\vwrt{\mwrs{\mathsf{L}_1}{\wfmp}{s}}{\wfmp}, \ldots, \vwrt{\mwrs{\mathsf{L}_m}{\wfmp}{s}}{\wfmp}\}$ follows immediately and implies that $\wfmp$ is a model of $\mathsf{P}$. To see that it is minimum, assume there exists a model $\mathcal{M}$ of the higher-order program $\mathsf{P}$, distinct from $\wfmp$, which does not satisfy $\wfmp \ale \mathcal{M}$.
Then the valuation function $\vwrt{\cdot}{\mathcal{M}}$ of $\mathcal{M}$ defines a (first-order) interpretation $M$ for the ground instantiation $\mathsf{Gr(P)}$ of $\mathsf{P}$, if, for every ground atom $\mathsf{A}$, we take $M(\mathsf{A}) = \vwrt{\mathsf{A}}{\mathcal{M}}$. It is obvious that $\wgp\not\ale M$, since $M(\mathsf{A}) = \vwrt{\mathsf{A}}{\mathcal{M}}$ and $\vwrt{\mathsf{A}}{\wfmp} = \wgp(\mathsf{A})$ for every ground atom $\mathsf{A}$ imply that $M\parallel v=\mathcal{M}\parallel v$ and $\wgp\parallel v = \wfmp\parallel v$ for every truth value $v\in \mathbb{V}$.
Also, $M$ is distinct from $\wgp$, since $M(\mathsf{B}) = \vwrt{\mathsf{B}}{\mathcal{M}} \neq \vwrt{\mathsf{B}}{\wfmp} = \wgp(\mathsf{B})$ for at least one ground atom $\mathsf{B}$. Next we demonstrate that $M$ is a model of the ground instantiation $\mathsf{Gr(P)}$ of $\mathsf{P}$, which is of course a contradiction, since $\wgp$ is the minimum model of $\mathsf{Gr(P)}$ and $\wgp\aleq M$ should hold. Indeed, every clause in $\mathsf{Gr(P)}$ is a ground instance of a clause $\mathsf{A} \leftarrow \mathsf{L}_1, \ldots, \mathsf{L}_m$ in $\mathsf{P}$ and is therefore of the form $\mathsf{A}\theta \leftarrow \mathsf{L}_1\theta, \ldots, \mathsf{L}_m\theta$ for some ground substitution $\theta$. Consider a Herbrand state $s$, such that $s(\mathsf{V}) = \theta(\mathsf{V})$ for every variable $\mathsf{V}$ in $\mathsf{P}$. Because we assumed the higher-order interpretation $\mathcal{M}$ to be a model of $\mathsf{P}$, we have that $\vwrt{\mwrs{\mathsf{A}}{\mathcal{M}}{s}}{\mathcal{M}}\geq min\{\vwrt{\mwrs{\mathsf{L}_1}{\mathcal{M}}{s}}{\mathcal{M}}, \ldots, \vwrt{\mwrs{\mathsf{L}_m}{\mathcal{M}}{s}}{\mathcal{M}}\}$. Again, it is easy to see that $\mathsf{A}\theta= \mwrs{\mathsf{A}}{\mathcal{M}}{s}$ and therefore $\vwrt{\mathsf{A}\theta}{\mathcal{M}}= \vwrt{\mwrs{\mathsf{A}}{\mathcal{M}}{s}}{\mathcal{M}}$. Similarly, $\vwrt{\mathsf{L}_i\theta}{\mathcal{M}}= \vwrt{\mwrs{\mathsf{L}_i}{\mathcal{M}}{s}}{\mathcal{M}}$ for all $i\leq m$. Additionally, by the construction of $M$, $\vwrt{\mathsf{A}\theta}{\mathcal{M}}= M(\mathsf{A}\theta)$ and $\vwrt{\mathsf{L}_i\theta}{\mathcal{M}}= M(\mathsf{L}_i\theta)$ for all $i\leq m$, so $M(\mathsf{A}\theta)\geq min\{M(\mathsf{L}_1\theta), \ldots, M(\mathsf{L}_m\theta)\}$, which implies that $M$ is a model of $\mathsf{Gr(P)}$. Because we reached a contradiction, $\wfmp$ must be the minimum model of $\mathsf{P}$.
\end{proof}

\section{Extensionality of the Proposed Semantics}\label{extensionality}\label{sec:ext_mon}
In this section we show that the infinite-valued model we defined in the previous section enjoys
the extensionality property, as this was defined in~\cite{Bezem99}. For an intuitive explanation
of this property, the reader should consult again the example of Section~\ref{intuitive}. In order
to formally define this notion, Bezem introduced~\cite{Bezem99,Bezem01} relations $\exeq[\mathit{val},\tau]$
over the set of expressions of a given type $\tau$ and under a given valuation function $\mathit{val}$. These relations
intuitively express extensional equality of type $\tau$. For the purposes of this paper only extensional equality of argument types will be needed, for which the formal definition is as follows:
%
\begin{defi}
Let $\mathcal{S}$ be a functional type structure and $\mathit{val}$ be a valuation function for $\mathcal{S}$. For every argument type $\rho$ we define the relations $\exeq[\mathit{val},\rho]$ on $\mo{\rho}$ as follows: Let $d, d'\in \mo{\rho}$; then $d \exeq[\mathit{val},\rho] d'$ if and only if
\begin{enumerate}
\item $\rho = \iota$ and $d=d'$, or
\item $\rho = o$ and $\mathit{val}(d) = \mathit{val}(d')$, or
\item $\rho = \rho'\rightarrow\pi$ and $d\,e \exeq[\mathit{val},\pi] d'\,e'$ for all $e, e'\in\mo{\rho'}$, such that $e\exeq[\mathit{val},\rho'] e'$ and $d\,e, d'\,e'$ are both defined.
\end{enumerate}
\end{defi}
One can easily verify that, for all $d, d'\in \mo{\rho_1\rightarrow\cdots\rightarrow\rho_n\rightarrow o}$, $e_1, e_1' \in \mo{\rho_1}$, \ldots, $e_n, e_n'\in \mo{\rho_n}$, if $d\exeq[\mathit{val},\rho_1\rightarrow\cdots\rightarrow\rho_n\rightarrow o] d'$, $e_1\exeq[\mathit{val},\rho_1] e_1'$, \ldots, $e_n\exeq[\mathit{val},\rho_n] e_n'$ and $d\ e_1\ \cdots\ e_n, d'\ e_1'\ \cdots\ e_n'$ are both defined, then $\mathit{val}(d\ e_1\ \cdots\ e_n) = \mathit{val}(d'\ e_1'\ \cdots\ e_n')$.

Generally, it is not guaranteed that such relations will be equivalence relations; rather they are partial equivalences (they are shown in \cite{Bezem99} to be symmetric and transitive). However, we are going to see that the minimum model of a program defines true equivalence relations for all types $\tau$.
\begin{defi}
Let $\mathsf{P}$ be a program and let $I$ be a Herbrand interpretation of $\mathsf{P}$. We say $I$ is \emph{extensional} if for all argument types $\rho$  the relations $\exeq[\mathit{val}_I, \rho]$ are reflexive, i.e. for all $\mathsf{E} \in \mo{\rho}$, it holds that $\mathsf{E} \exeq[\mathit{val}_I, \rho] \mathsf{E}$.
\end{defi}
\begin{thm}[Extensionality]
\label{thr:extensionality}
$\wfmp$ is extensional.
\end{thm}
\begin{proof}
Since the valuation function of $\wfmp$ is $\wgp$, essentially we need to show that $\mathsf{E}\exeq[\wgp,\rho] \mathsf{E}$, for every ground expression $\mathsf{E}$ of every argument type $\rho$. We perform an induction on the structure of $\rho$. For the base types $\iota$ and $o$ the statement holds by definition. 
For the induction step, we prove the statement for a predicate type $\pi=\rho_1\rightarrow\cdots\rightarrow\rho_m\rightarrow o$, assuming that it holds for all types simpler than $\pi$ (i.e., for the types $\rho_1,\ldots,\rho_m,o$ and, recursively, the types that are simpler than $\rho_1,\ldots,\rho_m$).
Let $\mathsf{A}$ be any atom of the following form: $\mathsf{A}$ is headed by a predicate constant and all variables in $vars(\mathsf{A})$ are of types simpler than $\pi$. Let $\theta, \theta'$ be  ground substitutions, such that $vars(\mathsf{A}) \subseteq dom(\theta),dom(\theta')$ and $\theta(\mathsf{V}) \exeq[\wgp,\rho] \theta'(\mathsf{V})$ for any $\mathsf{V}:\rho$ in $vars(\mathsf{A})$. We claim it suffices to show the following two properties $P_1(\alpha)$ and $P_2(\alpha)$, for all ordinals $\alpha$:
\newlist{indentdescription}{description}{10}
\setlist[indentdescription]{leftmargin=1em,labelindent=1em}
\begin{indentdescription}
\item[$P_1(\alpha)$] if $\vgr{\mathsf{A}\theta}{M_\alpha}= T_\alpha$ then $\vgr{\mathsf{A}\theta'}{\wgp}=T_\alpha$;
\item[$P_2(\alpha)$] if $\vgr{\mathsf{A}\theta}{M_{\alpha}}=F_\alpha$ then $\vgr{\mathsf{A}\theta'}{\wgp}=F_\alpha$.
\end{indentdescription}
To see why proving the above properties is enough to establish that $\mathsf{E}\exeq[\wgp, \pi] \mathsf{E}$, observe the following: first of all, we assumed that $\pi$ is of the form $\rho_1\rightarrow\cdots\rightarrow\rho_m\rightarrow o$, so if $\mathsf{V}_1:\rho_1, \ldots, \mathsf{V}_m:\rho_m$ are variables, then $\mathsf{E}\; \mathsf{V}_1\; \cdots\; \mathsf{V}_m$ is an atom of the form described above. Also, by Lemma \ref{lm:approx_alpha_eq} we have that $\vgr{\mathsf{E}\; \theta(\mathsf{V}_1) \; \cdots\; \theta(\mathsf{V}_m)}{\wgp}=T_\alpha$ if and only if $\vgr{\mathsf{E}\; \theta(\mathsf{V}_1) \; \cdots\; \theta(\mathsf{V}_m)}{M_\alpha}=T_\alpha$. If $P_1(\alpha)$ holds, the latter implies that $\vgr{\mathsf{E}\; \theta'(\mathsf{V}_1) \; \cdots\; \theta'(\mathsf{V}_m)}{\wgp}=T_\alpha$. Because the relations $\exeq[\wgp,\rho_i]$ are symmetric, $\theta$ and $\theta'$ are interchangeable. Therefore the same argument can be used to infer the reverse implication, i.e. $\vgr{\mathsf{E}\; \theta'(\mathsf{V}_1) \; \cdots\; \theta'(\mathsf{V}_m)}{\wgp}=T_\alpha\Rightarrow\vgr{\mathsf{E}\; \theta(\mathsf{V}_1) \; \cdots\; \theta(\mathsf{V}_m)}{\wgp}=T_\alpha$, thus resulting to an equivalence. If $P_2(\alpha)$ holds, the analogous equivalence can be shown for the value $F_\alpha$, in the same way.
Finally, the equivalence for the $0$ value follows by a simple elimination argument: if for example $\vgr{\mathsf{E}\; \theta(\mathsf{V}_1) \; \cdots\; \theta(\mathsf{V}_m)}{\wgp}=0$, we make the assumption that $\vgr{\mathsf{E}\; \theta'(\mathsf{V}_1) \; \cdots\; \theta'(\mathsf{V}_m)}{\wgp}=T_\alpha$ (respectively, $F_\alpha$) for some ordinal $\alpha$. Then, by Lemma \ref{lm:approx_alpha_eq}, $\vgr{\mathsf{E}\; \theta'(\mathsf{V}_1) \; \cdots\; \theta'(\mathsf{V}_m)}{M_\alpha}=T_\alpha$ (respectively, $F_\alpha$), so if property $P_1(\alpha)$ (respectively, $P_2(\alpha)$) holds, it gives us that $\vgr{\mathsf{E}\; \theta(\mathsf{V}_1) \; \cdots\; \theta(\mathsf{V}_m)}{\wgp}=T_\alpha$ (respectively, $F_\alpha$), which is a contradiction. It follows that $\vgr{\mathsf{E}\; \theta'(\mathsf{V}_1) \; \cdots\; \theta'(\mathsf{V}_m)}{\wgp}=0$. Again, we can show the reverse implication by the same argument.

We will proceed by a second induction on $\alpha$.
\begin{description}
\item[Second Induction Basis ($\alpha=0$)]
We have $M_0 = T^\omega_{\mathsf{P},0}(\emptyset)$. Observe that $\vgr{\mathsf{A}\theta}{T^\omega_{\mathsf{P},0}(\emptyset)}$ will evaluate to $T_0$ if and only if there exists some $n<\omega$ for which $\vgr{\mathsf{A}\theta}{T^n_{\mathsf{P}}(\emptyset)}=T_0$. On the other hand, it will evaluate to $F_0$ if and only if there does not exist a $n<\omega$ for which $\vgr{\mathsf{A}\theta}{T^n_{\mathsf{P}}(\emptyset)}\neq F_0$. Therefore, in order to prove $P_1(0)$ and $P_2(0)$, we first need to perform a third induction on $n$ and prove the following two properties:
\begin{indentdescription}
\item[$P'_1(0,n)$] if $\vgr{\mathsf{A}\theta}{T_{\mathsf{P}}^n(\emptyset)}= T_0$ then $\vgr{\mathsf{A}\theta'}{\wgp}=T_0$;
\item[$P'_2(0,n)$] if $\vgr{\mathsf{A}\theta}{T_{\mathsf{P}}^n(\emptyset)}>F_0$ then $\vgr{\mathsf{A}\theta'}{\wgp}>F_0$.
\end{indentdescription}

\item[Third Induction Basis ($n=0$)]
Both $P'_1(0,0)$ and $P'_2(0,0)$ hold vacuously, since $T_{\mathsf{P}}^0(\emptyset)=\emptyset$, i.e. the interpretation that assigns $F_0$ to every atom.

\item[Third Induction Step ($n+1$)]
First we show $P_1'(0,n+1)$, assuming that $P_1'(0,n)$ holds. If $\vgr{\mathsf{A}\theta}{T_{\mathsf{P}}^{n+1}(\emptyset)}=T_0$, then there exists a clause $\mathsf{A}\theta\leftarrow \mathsf{L}_1,\ldots, \mathsf{L}_k$ in $\mathsf{Gr(P)}$ such that for each $i\leq k$, $\vgr{\mathsf{L}_i}{T_{\mathsf{P}}^{n}(\emptyset)}=T_0$. This implies that each $\mathsf{L}_i$ is a positive literal, since a negative one cannot be assigned the value $T_0$ in any interpretation.
This clause is a ground instance of a clause $\mathsf{p}\,\mathsf{V}_1\,\cdots\,\mathsf{V}_r \leftarrow \mathsf{K}_1,\ldots, \mathsf{K}_k$ in the higher-order program and there exists a substitution $\theta''$, such that $(\mathsf{p}\,\mathsf{V}_1\,\cdots\,\mathsf{V}_r)\theta''=\mathsf{A}$ and, for any variable $\mathsf{V}\not\in \{\mathsf{V}_1,\ldots,\mathsf{V}_r\}$ appearing in the body of the clause, $\theta''(\mathsf{V})$ is an appropriate ground term, so that $\mathsf{L}_i = \mathsf{K}_i\theta''\theta$ for all $i\leq k$. Observe that the variables appearing in the clause $(\mathsf{p}\,\mathsf{V}_1\,\cdots\,\mathsf{V}_r)\theta'' \leftarrow \mathsf{K}_1\theta'',\ldots, \mathsf{K}_k\theta''$ are exactly the variables appearing in $\mathsf{A}$ and they are all of types simpler than $\pi$. We distinguish the following cases for each $\mathsf{K}_i\theta''$, $i\leq k$:
\begin{enumerate}
\item \emph{$\mathsf{K}_i\theta''$ is of the form $(\mathsf{E}_1 \approx \mathsf{E}_2)$:} By definition, $\vgr{\mathsf{L}_i}{T_{\mathsf{P}}^{n}(\emptyset)}=\vgr{\mathsf{K}_i\theta''\theta}{T_{\mathsf{P}}^{n}(\emptyset)}=T_0$ implies that $\mathsf{E}_1\theta = \mathsf{E}_2\theta$. Since $\mathsf{E}_1$ and $\mathsf{E}_2$ are expressions of type $\iota$, all variables in $\mathsf{E}_1$ and $\mathsf{E}_2$ are also of type $\iota$ and, because $\exeq[\wgp,\iota]$ is defined as equality, we will have $\mathsf{E}_1\theta = \mathsf{E}_1\theta'$ and $\mathsf{E}_2\theta=\mathsf{E}_2\theta'$. Therefore $\mathsf{E}_1\theta' = \mathsf{E}_2\theta'$ and $\wgp(\mathsf{K}_i\theta''\theta')=T_0$ will also hold.
\item \emph{$\mathsf{K}_i\theta''$ is an atom and starts with a predicate constant:} As we observed, the variables appearing in $\mathsf{K}_i\theta''$ are of types simpler than $\pi$. By the third induction hypothesis, $\mathsf{K}_i\theta''$ satisfies property $P_1'(0,n)$ and therefore $\vgr{\mathsf{L}_i}{T_{\mathsf{P}}^{n}(\emptyset)}  = \vgr{\mathsf{K}_i\theta''\theta}{T_{\mathsf{P}}^{n}(\emptyset)} = T_0$ implies that $\wgp(\mathsf{K}_i\theta''\theta')=T_0$.
\item \emph{$\mathsf{K}_i\theta''$ is an atom and starts with a predicate variable:} Let $\mathsf{K}_i\theta''=\mathsf{V}\,\mathsf{E}_1\,\cdots\,\mathsf{E}_{m'}$ for some $\mathsf{V}\in vars(\mathsf{A})$. Then $\mathsf{B}=\theta(\mathsf{V})\,\mathsf{E}_1\,\cdots\,\mathsf{E}_{m'}$ is an atom that begins with a predicate constant and, by $vars(\mathsf{K}_i\theta'') \subseteq vars(\mathsf{A})$, all of the variables of $\mathsf{B}$ are of types simpler than $\pi$. Also, $\vgr{\mathsf{B}\theta}{T_{\mathsf{P}}^{n}(\emptyset)} =\vgr{\mathsf{K}_i\theta''\theta}{T_{\mathsf{P}}^{n}(\emptyset)}=T_0$, which, by property $P'_1(0,n)$ yields that $\vgr{\mathsf{B}\theta'}{\wgp}=\vgr{\theta(\mathsf{V})\,\mathsf{E}_1\theta'\,\cdots\,\mathsf{E}_{m'}\theta'}{\wgp}=T_0$ (1). Observe that the types of all arguments of $\theta(\mathsf{V})$, i.e. the types of $\mathsf{E}_j\theta'$ for all $j\leq m'$, are simpler than the type of  $\mathsf{V}$ and consequently, since $\mathsf{V} \in vars(\mathsf{A})$, simpler than $\pi$. For each $j\leq m'$, let $\rho_j$ be the type of $\mathsf{E}_j$ and let $\rho$ be the type of $\mathsf{V}$; by the first induction hypothesis, $\mathsf{E}_j\theta'\exeq[\wgp,\rho_j] \mathsf{E}_j\theta'$. Moreover, by assumption we have that $\theta(\mathsf{V}) \exeq[\wgp,\rho] \theta'(\mathsf{V})$. Then, by definition $\vgr{\theta(\mathsf{V})\,\mathsf{E}_1\theta'\,\cdots\,\mathsf{E}_{m'}\theta'}{\wgp}=\vgr{\theta'(\mathsf{V})\,\mathsf{E}_1\theta'\,\cdots\,\mathsf{E}_{m'}\theta'}{\wgp}$ and, by (1), $\vgr{\theta'(\mathsf{V})\,\mathsf{E}_1\theta'\,\cdots\,\mathsf{E}_{m'}\theta'}{\wgp}=T_0$.
\end{enumerate}
In conclusion, the clause $\mathsf{A}\theta' \leftarrow \mathsf{K}_1\theta''\theta',\ldots, \mathsf{K}_k\theta''\theta'$ is in $\mathsf{Gr(P)}$ and for each $i\leq k$, we have $\vgr{\mathsf{K}_i\theta''\theta'}{\wgp}=T_0$, therefore $\vgr{\mathsf{A}\theta'}{\wgp}=T_0$ must also hold.

This concludes the proof for $P'_1(0,n)$. Next we prove $P'_2(0,n+1)$, assuming $P'_2(0,n)$ holds. If $\vgr{\mathsf{A}\theta}{T_{\mathsf{P}}^{n+1}(\emptyset)}>F_0$, then there exists a clause $\mathsf{A}\theta\leftarrow \mathsf{L}_1,\ldots, \mathsf{L}_k$ in $\mathsf{Gr(P)}$ such that for each $i\leq k$, $\vgr{\mathsf{L}_i}{T_{\mathsf{P}}^{n}(\emptyset)}>F_0$.
This clause is a ground instance of a clause $\mathsf{p}\,\mathsf{V}_1\,\cdots\,\mathsf{V}_r \leftarrow \mathsf{K}_1,\ldots, \mathsf{K}_k$ in the higher-order program and there exists a substitution $\theta''$, such that $(\mathsf{p}\,\mathsf{V}_1\,\cdots\,\mathsf{V}_r)\theta''=\mathsf{A}$ and, for any variable $\mathsf{V}\not\in \{\mathsf{V}_1,\ldots,\mathsf{V}_r\}$ appearing in the body of the clause, $\theta''(\mathsf{V})$ is an appropriate ground term, so that $\mathsf{L}_i = \mathsf{K}_i\theta''\theta$ for all $i\leq k$. Observe that the variables appearing in the clause $(\mathsf{p}\,\mathsf{V}_1\,\cdots\,\mathsf{V}_r)\theta'' \leftarrow \mathsf{K}_1\theta'',\ldots, \mathsf{K}_k\theta''$ are exactly the variables appearing in $\mathsf{A}$ and they are all of types simpler than $\pi$. We can distinguish the following cases for each $\mathsf{K}_i\theta''$, $i\leq k$:
\begin{enumerate}
\item \emph{$\mathsf{K}_i\theta''$ is a positive literal:} A positive literal may take one of the three forms that we examined in our proof for property $P'_1(0,n+1)$. We can show that $\vgr{\mathsf{K}_i\theta''\theta'}{\wgp}> F_0$ by the same arguments we used in each case.
\item \emph{$\mathsf{K}_i\theta''$ is a negative literal:}
A negative literal cannot be assigned the value $F_0$ in any interpretation. Therefore $\vgr{\mathsf{K}_i\theta''\theta'}{\wgp}>F_0$ holds by definition.
\end{enumerate}
In conclusion, the clause $\mathsf{A}\theta' \leftarrow \mathsf{K}_1\theta''\theta',\ldots, \mathsf{K}_k\theta''\theta'$ is in $\mathsf{Gr(P)}$ and for each $i\leq k$, we have $\vgr{\mathsf{K}_i\theta''\theta'}{\wgp}>F_0$, therefore $\vgr{\mathsf{A}\theta'}{\wgp}>F_0$ must also hold.

This concludes the proof for $P'_2(0,n)$.
We will now use properties $P'_1(0,n)$ and $P'_2(0,n)$ in order to show $P_1(0)$ and $P_2(0)$. By definition, if $\vgr{\mathsf{A}\theta}{M_0}= \vgr{\mathsf{A}\theta}{T^\omega_{\mathsf{P},0}(\emptyset)}=T_0$, then there exists some $n<\omega$ such that $\vgr{\mathsf{A}\theta}{T_{\mathsf{P}}^{n}(\emptyset)}=T_0$. Applying $P'_1(0,n)$ to $\mathsf{A}\theta$ we immediately conclude that $\vgr{\mathsf{A}\theta'}{\wgp}= T_0$, which establishes property $P_1(0)$. Now let $\vgr{\mathsf{A}\theta}{M_0}=F_0$ and assume $\vgr{\mathsf{A}\theta'}{\wgp}\neq F_0$. By Lemma \ref{lm:approx_alpha_eq}, the latter can only hold if $\vgr{\mathsf{A}\theta'}{M_0}=\vgr{\mathsf{A}\theta'}{T_{\mathsf{P},0}^{\omega}(\emptyset)} \neq F_0$ and this, in turn, means that there exists at least one $n<\omega$ such that $\vgr{\mathsf{A}\theta'}{T_{\mathsf{P}}^{n}(\emptyset)}>F_0$. Then, reversing the roles of $\theta$ and $\theta'$, we can apply property $P'_2(0,n)$ to $\mathsf{A}\theta'$ and conclude that $\vgr{\mathsf{A}\theta}{\wgp}>F_0$, which, again by Lemma \ref{lm:approx_alpha_eq}, contradicts $\vgr{\mathsf{A}\theta}{M_0}=F_0$. Therefore it must be $\vgr{\mathsf{A}\theta'}{\wgp}= F_0$.

\item[Second Induction Step]
Now we prove properties $P_1(\alpha)$ and $P_2(\alpha)$ for an arbitrary countable ordinal $\alpha$, assuming that $P_1(\beta)$ and $P_2(\beta)$ hold for all $\beta<\alpha$.

We have $M_\alpha = T^\omega_{\mathsf{P},\alpha}(\baseap{\alpha})$. Again, we first perform a third induction on $n$ and prove two auxilary properties, as follows:
\begin{indentdescription}
\item[$P'_1(\alpha,n)$] if $\vgr{\mathsf{A}\theta}{T_{\mathsf{P}}^n(\baseap{\alpha})}\geq T_\alpha$ then $\vgr{\mathsf{A}\theta'}{\wgp}\geq T_\alpha$;
\item[$P'_2(\alpha,n)$] if $\vgr{\mathsf{A}\theta}{T_{\mathsf{P}}^n(\baseap{\alpha})}>F_\alpha$ then $\vgr{\mathsf{A}\theta'}{\wgp}>F_\alpha$.
\end{indentdescription}

\item[Third Induction Basis ($n=0$)]
We have $T_{\mathsf{P}}^0(\baseap{\alpha})=\baseap{\alpha}$. Observe that, whether $\alpha$ is a successor or a limit ordinal, $\baseap{\alpha}$ does not assign to any atom
the value $T_\alpha$ or any value that is greater than $F_\alpha$ and smaller than $T_\alpha$. So, whether we assume $\vgr{\mathsf{A}\theta}{T_{\mathsf{P}}^0(\baseap{\alpha})}\geq T_\alpha$ or $\vgr{\mathsf{A}\theta}{T_{\mathsf{P}}^0(\baseap{\alpha})}> F_\alpha$, it must be $\vgr{\mathsf{A}\theta}{T_{\mathsf{P}}^0(\baseap{\alpha})}=T_\beta$ for some ordinal $\beta< \alpha$. By Lemma \ref{lm:TPn-and-TPomega}, $\vgr{\mathsf{A}\theta}{M_\alpha}=T_\beta$ and so, by Lemma \ref{lm:approx_alpha_eq}, $\vgr{\mathsf{A}\theta}{M_\beta}=T_\beta$. Then, by the second induction hypothesis, property $P_1(\beta)$ holds and yields $\vgr{\mathsf{A}\theta'}{\wgp}=T_\beta$, which implies both $\vgr{\mathsf{A}\theta'}{\wgp}\geq T_\alpha$ and $\vgr{\mathsf{A}\theta'}{\wgp}>F_\alpha$. Therefore property $P'_1(\alpha,0)$ and property $P'_2(\alpha,0)$ also hold.

\item[Third Induction Step ($n+1$)]
First we show $P_1'(\alpha,n+1)$, assuming that $P_1'(\alpha,n)$ holds. If $\vgr{\mathsf{A}\theta}{T_{\mathsf{P}}^{n+1}(\baseap{\alpha})}\geq T_\alpha$, then there exists a clause $\mathsf{A}\theta\leftarrow \mathsf{L}_1,\ldots, \mathsf{L}_k$ in $\mathsf{Gr(P)}$ such that for each $i\leq k$, $\vgr{\mathsf{L}_i}{T_{\mathsf{P}}^{n}(\baseap{\alpha})}\geq T_\alpha$.
This clause is a ground instance of a clause $\mathsf{p}\,\mathsf{V}_1\,\cdots\,\mathsf{V}_r \leftarrow \mathsf{K}_1,\ldots, \mathsf{K}_k$ in the higher-order program and there exists a substitution $\theta''$, such that $(\mathsf{p}\,\mathsf{V}_1\,\cdots\,\mathsf{V}_r)\theta''=\mathsf{A}$ and, for any variable $\mathsf{V}\not\in \{\mathsf{V}_1,\ldots,\mathsf{V}_r\}$ appearing in the body of the clause, $\theta''(\mathsf{V})$ is an appropriate ground term, so that $\mathsf{L}_i = \mathsf{K}_i\theta''\theta$ for all $i\leq k$. As we observed earlier, the variables appearing in the clause $(\mathsf{p}\,\mathsf{V}_1\,\cdots\,\mathsf{V}_r)\theta'' \leftarrow \mathsf{K}_1\theta'',\ldots, \mathsf{K}_k\theta''$ are exactly the variables appearing in $\mathsf{A}$ and they are all of types simpler than $\pi$.
We can distinguish the following cases for each $\mathsf{K}_i\theta''$:
\begin{enumerate}
\item \emph{$\mathsf{K}_i\theta''$ is a positive literal:} A positive literal may take one of the three forms that we examined in our proof for property $P'_1(0,n+1)$. We can show that $\vgr{\mathsf{K}_i\theta''\theta'}{\wgp}\geq T_\alpha$ by the same arguments we used in each case.
\item \emph{$\mathsf{K}_i\theta''$ is a negative literal and its atom starts with a predicate constant:}
Let $\mathsf{K}_i\theta''$ be of the form $\pnot \mathsf{B}$, where $\mathsf{B}$ is an atom that starts with a predicate constant. It is $\vgr{\pnot\mathsf{B}\theta}{T_{\mathsf{P}}^{n}(\baseap{\alpha})}=\vgr{\mathsf{K}_i\theta''\theta}{T_{\mathsf{P}}^{n}(\baseap{\alpha})}=\vgr{\mathsf{L}_i}{T_{\mathsf{P}}^{n}(\baseap{\alpha})}\geq T_\alpha$. Then $\vgr{\mathsf{B}\theta}{T_{\mathsf{P}}^{n}(\baseap{\alpha})}< F_\alpha$, i.e. $\vgr{\mathsf{B}\theta}{T_{\mathsf{P}}^{n}(\baseap{\alpha})}= F_\beta$ for some ordinal $\beta<\alpha$. By Lemma \ref{lm:TPn-and-TPomega}, $\vgr{\mathsf{B}\theta}{M_{\alpha}}= F_\beta$ and thus, by Lemma \ref{lm:approx_alpha_eq}, $\vgr{\mathsf{B}\theta}{M_\beta}= F_\beta$. By $vars(\mathsf{K}_i\theta'') \subseteq vars(\mathsf{A})$, all the variables of $\mathsf{B}$ are of types simpler than $\pi$, so we can apply the second induction hypothesis, in particular property $P_2(\beta)$, to $\mathsf{B}\theta$ and conclude that $\vgr{\mathsf{B}\theta'}{\wgp}= F_\beta$. Then $\vgr{\mathsf{K}_i\theta''\theta'}{\wgp}=\vgr{\pnot \mathsf{B}\theta'}{\wgp}= T_{\beta+1}\geq T_\alpha$.
\item \emph{$\mathsf{K}_i\theta''$ is a negative literal and its atom starts with a predicate variable:}
Let $\mathsf{K}_i\theta''=\pnot (\mathsf{V}\,\mathsf{E}_1\,\cdots\,\mathsf{E}_{m'})$ for some $\mathsf{V}\in vars(\mathsf{A})$. Then $\mathsf{B}=\theta(\mathsf{V})\,\mathsf{E}_1\,\cdots\,\mathsf{E}_{m'}$ is an atom that begins with a predicate constant and, by $vars(\mathsf{K}_i\theta'') \subseteq vars(\mathsf{A})$, all the variables of $\mathsf{B}$ are of types simpler than $\pi$. Also, $\vgr{\pnot\mathsf{B}\theta}{T_{\mathsf{P}}^{n}(\baseap{\alpha})} =\vgr{\mathsf{K}_i\theta''\theta}{T_{\mathsf{P}}^{n}(\baseap{\alpha})}=\vgr{\mathsf{L}_i}{T_{\mathsf{P}}^{n}(\baseap{\alpha})}\geq T_\alpha$. Then $\vgr{\mathsf{B}\theta}{T_{\mathsf{P}}^{n}(\baseap{\alpha})}< F_{\alpha}$, i.e. $\vgr{\mathsf{B}\theta}{T_{\mathsf{P}}^{n}(\baseap{\alpha})}=F_{\beta}$ for some ordinal $\beta< \alpha$. By Lemma~\ref{lm:TPn-and-TPomega}, $\vgr{\mathsf{B}\theta}{M_{\alpha}}= F_\beta$ and thus, by Lemma~\ref{lm:approx_alpha_eq}, $\vgr{\mathsf{B}\theta}{M_\beta}= F_\beta$. By the second induction hypothesis, property $P_2(\beta)$ gives us that $\vgr{\mathsf{B}\theta'}{\wgp}=\vgr{\theta(\mathsf{V})\,\mathsf{E}_1\theta'\,\cdots\,\mathsf{E}_{m'}\theta'}{\wgp}=F_\beta$~(1). The types of all arguments of $\theta(\mathsf{V})$, i.e. the types of $\mathsf{E}_j\theta'$ for all $j\leq m'$, are simpler than the type of  $\mathsf{V}$ and consequently, since $\mathsf{V} \in vars(\mathsf{A})$, simpler than $\pi$. For each $j\leq m'$, let $\rho_j$ be the type of $\mathsf{E}_j$ and let $\rho$ be the type of $\mathsf{V}$; by the first induction hypothesis, $\mathsf{E}_j\theta'\exeq[\wgp,\rho_j] \mathsf{E}_j\theta'$. Moreover, by assumption we have that $\theta(\mathsf{V}) \exeq[\wgp,\rho] \theta'(\mathsf{V})$. Then, by definition $\vgr{\theta(\mathsf{V})\,\mathsf{E}_1\theta'\,\cdots\,\mathsf{E}_{m'}\theta'}{\wgp}=\vgr{\theta'(\mathsf{V})\,\mathsf{E}_1\theta'\,\cdots\,\mathsf{E}_{m'}\theta'}{\wgp}$ and, by (1), $\vgr{\theta'(\mathsf{V})\,\mathsf{E}_1\theta'\,\cdots\,\mathsf{E}_{m'}\theta'}{\wgp}=F_\beta$. Therefore, it follows that $\vgr{\mathsf{K}_i\theta''\theta'}{\wgp}=\vgr{\pnot (\theta'(\mathsf{V})\,\mathsf{E}_1\theta'\,\cdots\,\mathsf{E}_{m'}\theta')}{\wgp}=T_{\beta+1}\geq T_\alpha$.
\end{enumerate}
We can conclude that $\vgr{\mathsf{A}\theta'}{\wgp}\geq T_\alpha$, as the clause $\mathsf{A}\theta' \leftarrow \mathsf{K}_1\theta''\theta',\ldots, \mathsf{K}_k\theta''\theta'$ is in $\mathsf{Gr(P)}$ and we have shown that, for each $i\leq k$, $\vgr{\mathsf{K}_i\theta''\theta'}{\wgp}\geq T_\alpha$.

This concludes the proof of $P'_1(\alpha,n)$. Next we prove $P'_2(\alpha,n+1)$, assuming $P'_2(\alpha,n)$ holds. If $\vgr{\mathsf{A}\theta}{T_{\mathsf{P}}^{n+1}(\baseap{\alpha})}>F_\alpha$, then there exists a clause $\mathsf{A}\theta\leftarrow \mathsf{L}_1,\ldots, \mathsf{L}_k$ in $\mathsf{Gr(P)}$ such that for each $i\leq k$, $\vgr{\mathsf{L}_i}{T_{\mathsf{P}}^{n}(\baseap{\alpha})}>F_\alpha$.
This clause is a ground instance of a clause $\mathsf{p}\,\mathsf{V}_1\,\cdots\,\mathsf{V}_r \leftarrow \mathsf{K}_1,\ldots, \mathsf{K}_k$ in the higher-order program and there exists a substitution $\theta''$such that $(\mathsf{p}\,\mathsf{V}_1\,\cdots\,\mathsf{V}_r)\theta''=\mathsf{A}$ and, for any variable $\mathsf{V}\not\in \{\mathsf{V}_1,\ldots,\mathsf{V}_r\}$ appearing in the body of the clause, $\theta''(\mathsf{V})$ is an appropriate ground term, so that $\mathsf{L}_i = \mathsf{K}_i\theta''\theta$ for all $i\leq k$. Again we observe that the variables appearing in the clause $(\mathsf{p}\,\mathsf{V}_1\,\cdots\,\mathsf{V}_r)\theta'' \leftarrow \mathsf{K}_1\theta'',\ldots, \mathsf{K}_k\theta''$ are exactly the variables appearing in $\mathsf{A}$, which are all of types simpler than $\pi$, and distinguish the following cases for each $\mathsf{K}_i\theta''$:
\begin{enumerate}
\item \emph{$\mathsf{K}_i\theta''$ is a positive literal:} A positive literal may take one of the three forms that we examined in our proof for property $P'_1(0,n+1)$. We can show that $\vgr{\mathsf{K}_i\theta''\theta'}{\wgp}> F_\alpha$ by the same arguments we used in each case.
\item \emph{$\mathsf{K}_i\theta''$ is a negative literal and its atom starts with a predicate constant:}
Let $\mathsf{K}_i\theta''$ be of the form $\pnot \mathsf{B}$, where $\mathsf{B}$ is an atom that starts with a predicate constant and, by $vars(\mathsf{K}_i\theta'') \subseteq vars(\mathsf{A})$, all the variables of $\mathsf{B}$ are of types simpler than $\pi$. It is $\vgr{\pnot\mathsf{B}\theta}{T_{\mathsf{P}}^{n}(\baseap{\alpha})} =\vgr{\mathsf{K}_i\theta''\theta}{T_{\mathsf{P}}^{n}(\baseap{\alpha})}=\vgr{\mathsf{L}_i}{T_{\mathsf{P}}^{n}(\baseap{\alpha})}> F_\alpha$ and we claim that $\vgr{\pnot \mathsf{B}\theta'}{\wgp}> F_\alpha$. For the sake of contradiction, assume that $\vgr{\pnot\mathsf{B}\theta'}{\wgp}\leq F_\alpha$. Then $\vgr{\mathsf{B}\theta'}{\wgp}> T_\alpha$, i.e. $\vgr{\mathsf{B}\theta'}{\wgp}= T_\beta$ for some ordinal $\beta<\alpha$. Therefore, Lemma \ref{lm:approx_alpha_eq} implies that $\vgr{\mathsf{B}\theta'}{M_\beta}= T_\beta$. This means that, if we reverse the roles of $\theta$ and $\theta'$, we can apply the second induction hypothesis to $\mathsf{B}\theta'$ and use property $P_1(\beta)$ to conclude that $\vgr{\mathsf{B}\theta}{\wgp}= T_\beta$. By Lemma \ref{lm:approx_alpha_eq}, this implies that $\vgr{\mathsf{B}\theta}{M_{\alpha}}= T_\beta$ and, by Lemma \ref{lm:TPn-and-TPomega}, $\vgr{\mathsf{B}\theta}{T_{\mathsf{P}}^{n}(\baseap{\alpha})}=T_\beta$. Then $\vgr{\pnot \mathsf{B}\theta}{T_{\mathsf{P}}^{n}(\baseap{\alpha})}= F_{\beta+1}\leq F_\alpha$, which is obviously a contradiction. So it must be $\vgr{\mathsf{K}_i\theta''\theta'}{\wgp}=\vgr{\pnot \mathsf{B}\theta'}{\wgp}> F_\alpha$.
\item \emph{$\mathsf{K}_i\theta''$ is a negative literal and its atom starts with a predicate variable:}
Let $\mathsf{K}_i\theta''=\pnot (\mathsf{V}\,\mathsf{E}_1\,\cdots\,\mathsf{E}_{m'})$ for some $\mathsf{V}\in vars(\mathsf{A})$. Then $\mathsf{B}=\theta(\mathsf{V})\,\mathsf{E}_1\,\cdots\,\mathsf{E}_{m'}$ is an atom that begins with a predicate constant and, by $vars(\mathsf{K}_i\theta'') \subseteq vars(\mathsf{A})$, all the variables of $\mathsf{B}$ are of types simpler than $\pi$. Also, $\vgr{\pnot\mathsf{B}\theta}{T_{\mathsf{P}}^{n}(\baseap{\alpha})} =\vgr{\mathsf{K}_i\theta''\theta}{T_{\mathsf{P}}^{n}(\baseap{\alpha})}=\vgr{\mathsf{L}_i}{T_{\mathsf{P}}^{n}(\baseap{\alpha})}> F_\alpha$. We claim that $\vgr{\pnot\mathsf{B}\theta'}{\wgp}>F_{\alpha}$. Again, assume that this is not so, i.e. $\vgr{\pnot\mathsf{B}\theta'}{\wgp}\leq F_{\alpha}$; then $\vgr{\mathsf{B}\theta'}{\wgp}=T_\beta> T_{\alpha}$ for some ordinal $\beta < \alpha$. By Lemma \ref{lm:approx_alpha_eq}, $\vgr{\mathsf{B}\theta'}{M_\beta}=T_\beta$ and the second induction hypothesis applies. So if we reverse the roles of $\theta$ and $\theta'$ property $P_1(\beta)$ gives us that $\vgr{\mathsf{B}\theta}{\wgp}=T_\beta$. By Lemma \ref{lm:approx_alpha_eq}, $\vgr{\mathsf{B}\theta}{M_{\alpha}}= T_\beta$ and thus, by Lemma \ref{lm:TPn-and-TPomega}, $\vgr{\mathsf{B}\theta}{T_{\mathsf{P}}^{n}(\baseap{\alpha})}= T_\beta$. This is a contradiction, as it implies that $\vgr{\pnot\mathsf{B}\theta}{T_{\mathsf{P}}^{n}(\baseap{\alpha})}=F_{\beta+1}\leq F_\alpha$. Therefore it must hold that $\vgr{\pnot\mathsf{B}\theta'}{\wgp}=\vgr{\pnot(\theta(\mathsf{V})\,\mathsf{E}_1\theta'\,\cdots\,\mathsf{E}_{m'}\theta')}{\wgp}> F_{\alpha}$~(1). The types of all arguments of $\theta(\mathsf{V})$, i.e. the types of $\mathsf{E}_j\theta'$ for all $j\leq m'$, are simpler than the type of $\mathsf{V}$ and therefore simpler than $\pi$. For each $j\leq m'$, let $\rho_j$ be the type of $\mathsf{E}_j$ and let $\rho$ be the type of $\mathsf{V}$; by the first induction hypothesis, $\mathsf{E}_j\theta'\exeq[\wgp,\rho_j] \mathsf{E}_j\theta'$ and by assumption $\theta(\mathsf{V}) \exeq[\wgp,\rho] \theta'(\mathsf{V})$. Then, by definition $\vgr{\theta(\mathsf{V})\,\mathsf{E}_1\theta'\,\cdots\,\mathsf{E}_{m'}\theta'}{\wgp}=\vgr{\theta'(\mathsf{V})\,\mathsf{E}_1\theta'\,\cdots\,\mathsf{E}_{m'}\theta'}{\wgp}$ and, consequently, $\vgr{\pnot(\theta(\mathsf{V})\,\mathsf{E}_1\theta'\,\cdots\,\mathsf{E}_{m'}\theta')}{\wgp}=\vgr{\pnot(\theta'(\mathsf{V})\,\mathsf{E}_1\theta'\,\cdots\,\mathsf{E}_{m'}\theta')}{\wgp}$. Therefore by (1), $\vgr{\pnot(\theta'(\mathsf{V})\,\mathsf{E}_1\theta'\,\cdots\,\mathsf{E}_{m'}\theta')}{\wgp}>F_\alpha$
and so we can conclude that
 $\vgr{\mathsf{K}_i\theta''\theta'}{\wgp}=\vgr{\pnot (\theta'(\mathsf{V})\,\mathsf{E}_1\theta'\,\cdots\,\mathsf{E}_{m'}\theta')}{\wgp}> F_\alpha$.
\end{enumerate}
Observe that the clause $\mathsf{A}\theta' \leftarrow \mathsf{K}_1\theta''\theta',\ldots, \mathsf{K}_k\theta''\theta'$ is in $\mathsf{Gr(P)}$ and we have shown that, for each $i\leq k$, $\vgr{\mathsf{K}_i\theta''\theta'}{\wgp}>F_\alpha$. So $\vgr{\mathsf{A}\theta'}{\wgp}> F_\alpha$ must also hold.

This concludes the proof for $P'_2(\alpha,n)$.
We will now use properties $P'_1(\alpha,n)$ and $P'_2(\alpha,n)$ in order to show $P_1(\alpha)$ and $P_2(\alpha)$. By definition, if $\vgr{\mathsf{A}\theta}{M_\alpha}= \vgr{\mathsf{A}\theta}{T^\omega_{\mathsf{P},\alpha}(\baseap{\alpha})}=T_\alpha$, then there exists some $n<\omega$ such that $\vgr{\mathsf{A}\theta}{T_{\mathsf{P}}^{n}(\baseap{\alpha})}=T_\alpha$. As we have shown above, property $P'_1(\alpha,n)$ yields $\vgr{\mathsf{A}\theta'}{\wgp}\geq T_\alpha$. However, if it was $\vgr{\mathsf{A}\theta'}{\wgp}=T_\beta > T_\alpha$ for some ordinal $\beta<\alpha$, then, by Lemma \ref{lm:approx_alpha_eq} we would also have $\vgr{\mathsf{A}\theta'}{M_\beta}=T_\beta$. By the second induction hypothesis we would be able to apply property $P_1(\beta)$ to $\mathsf{A}\theta'$ and infer that $\vgr{\mathsf{A}\theta}{\wgp}=T_\beta$, which, again by Lemma \ref{lm:approx_alpha_eq}, contradicts $\vgr{\mathsf{A}\theta}{M_\alpha}=T_\alpha$. So $\vgr{\mathsf{A}\theta'}{\wgp}$ can only be equal to $T_\alpha$ and property $P_1(\alpha)$ holds. Now let $\vgr{\mathsf{A}\theta}{M_\alpha}=F_\alpha$ and assume $\vgr{\mathsf{A}\theta'}{\wgp}\neq F_\alpha$, i.e. either $\vgr{\mathsf{A}\theta'}{\wgp}< F_\alpha$ or $\vgr{\mathsf{A}\theta'}{\wgp}> F_\alpha$. In the first case, $\vgr{\mathsf{A}\theta'}{\wgp}= F_\beta$ and, by Lemma \ref{lm:approx_alpha_eq}, $\vgr{\mathsf{A}\theta'}{M_\beta}= F_\beta$ for some $\beta<\alpha$. By the second induction hypothesis, property $P_2(\beta)$ gives us that $\vgr{\mathsf{A}\theta}{\wgp}= F_\beta<F_\alpha$. In the second case, Lemma \ref{lm:approx_alpha_eq} implies that $\vgr{\mathsf{A}\theta'}{M_\alpha}=\vgr{\mathsf{A}\theta'}{T_{\mathsf{P},\alpha}^{\omega}(\baseap{\alpha})} > F_\alpha$ and this, in turn, means that there exists at least one $n<\omega$ such that $\vgr{\mathsf{A}\theta'}{T_{\mathsf{P}}^{n}(\baseap{\alpha})}>F_\alpha$. Then, reversing the roles of $\theta$ and $\theta'$, we can apply property $P'_2(\alpha,n)$ to $\mathsf{A}\theta'$ and conclude that $\vgr{\mathsf{A}\theta}{\wgp}>F_\alpha$. In both cases, our conclusion constitutes a contradiction, because, by Lemma \ref{lm:approx_alpha_eq}, $\vgr{\mathsf{A}\theta}{M_\alpha}=F_\alpha$ implies that $\vgr{\mathsf{A}\theta}{\wgp}=F_\alpha$. Therefore it must also be $\vgr{\mathsf{A}\theta'}{\wgp}= F_\alpha$ and this proves property $P_2(\alpha)$.
\qedhere
\end{description}
\end{proof}

\section{Stratified and Locally Stratified Programs}\label{stratification_section}
In this section we define the notions of {\em stratified} and {\em locally stratified} programs and
argue that atoms of such programs never obtain the truth value 0 under the proposed semantics.
The notion of local stratification is a straightforward generalization of the corresponding
notion for classical (first-order) logic programs. However, the notion of stratification is
a genuine extension of the corresponding notion for first-order programs.
\begin{defi}\label{locally-stratified}
A program $\mathsf{P}$ is called {\em locally stratified} if and only if it is possible to decompose the
Herbrand base $U_{\mathsf{P},o}$ of $\mathsf{P}$ into disjoint sets (called {\em strata})
$S_1,S_2,\ldots,S_\alpha,\ldots, \alpha<\gamma$, where $\gamma$ is a countable ordinal, such that
for every clause $\mathsf{H} \leftarrow \mathsf{A}_1,\ldots,\mathsf{A}_m,\pnot \mathsf{B}_1,\ldots,\pnot \mathsf{B}_n$
in $\mathsf{Gr}(\mathsf{P})$, we have that for every $i\leq m$, $\textit{stratum}(\mathsf{A}_i) \leq \textit{stratum}(\mathsf{H})$ and
for every $i\leq n$, $\textit{stratum}(\mathsf{B}_i) < \textit{stratum}(\mathsf{H})$, where $\textit{stratum}$ is a function
such that $\textit{stratum}(\mathsf{C}) = \beta$,
if the atom $\mathsf{C}\in U_{\mathsf{P},o}$ belongs to $S_\beta$, and $\textit{stratum}(\mathsf{C}) = 0$, if  $\mathsf{C}\not\in U_{\mathsf{P},o}$ and is of the form $(\mathsf{E}_1\approx\mathsf{E}_2)$.
\end{defi}
All atoms in the minimum Herbrand model of a locally stratified program have non-zero values:
\begin{lem}
Let $\mathsf{P}$ be a locally stratified logic program. Then, for every atom $\mathsf{A} \in U_{\mathsf{P},o}$
it holds $\wfmp(\mathsf{A}) \neq 0$.
\end{lem}
\begin{proof}
Theorem~\ref{collapses} implies that the infinite-valued model $\mgp$ of the ground instantiation of $\mathsf{P}$ assigns the truth value $0$ to an atom iff this atom is assigned the truth value $0$ by the well-founded model of $\mathsf{Gr}(\mathsf{P})$. Notice now that, by Definition~\ref{locally-stratified}, if $\mathsf{P}$ is a locally stratified higher-order program, then $\mathsf{Gr(P)}$ is in turn a locally stratified propositional program (having exactly the same local stratification as $\mathsf{P}$). Recall that the well-founded model of a locally stratified propositional program does not assign the truth value $0$ to any atom~\cite{GelderRS91}, so neither does $\mgp$ or, consequently, $\wfmp$.
\end{proof}

Since Definition~\ref{locally-stratified} generalizes the corresponding one for classical
logic programs, the undecidability result~\cite{CholakB94} for detecting whether a given
program is locally stratified, extends directly to the higher-order case.
\begin{lem}
The problem of determining whether a given logic program $\mathsf{P}$ is locally stratified, is undecidable.
\end{lem}
However, there exists a notion of stratification for higher-order logic programs that is decidable
and has as a special case the stratification for classical logic programs~\cite{AptBW88}. In the following definition, a predicate type $\pi$ is understood to be \emph{greater than} a second predicate type $\pi'$, if $\pi$ is of the form $\rho_1\rightarrow\cdots\rightarrow\rho_n\rightarrow\pi'$, where $n\geq1$.

\begin{defi}\label{stratified}
A program $\mathsf{P}$ is called {\em stratified} if and only if it is possible to decompose the
set of all predicate constants that appear in $\mathsf{P}$ into a finite number $r$ of disjoint sets (called {\em strata})
$S_1,S_2,\ldots,S_r$, such that for every clause
$\mathsf{H} \leftarrow \mathsf{A}_1,\ldots,\mathsf{A}_m,\pnot \mathsf{B}_1,\ldots,\pnot \mathsf{B}_n$
in $\mathsf{P}$, where the predicate constant of $\mathsf{H}$ is $\mathsf{p}$, we have:
\begin{enumerate}
\item for every $i\leq m$, if $\mathsf{A}_i$ is a term that starts with a predicate constant $\mathsf{q}$, then
      $\textit{stratum}(\mathsf{q}) \leq \textit{stratum}(\mathsf{p})$;

\item for every $i\leq m$, if $\mathsf{A}_i$ is a term that starts with a predicate variable $\mathsf{Q}$, then
      for all predicate constants $\mathsf{q}$ that appear in $\mathsf{P}$ such that the type of
      $\mathsf{q}$ is greater than or equal to the type of $\mathsf{Q}$, it holds
      $\textit{stratum}(\mathsf{q}) \leq \textit{stratum}(\mathsf{p})$;

\item for every $i\leq n$, if $\mathsf{B}_i$ starts with a predicate constant $\mathsf{q}$, then
      $\textit{stratum}(\mathsf{q}) < \textit{stratum}(\mathsf{p})$;

\item for every $i\leq n$, if $\mathsf{B}_i$ starts with a predicate variable $\mathsf{Q}$, then
      for all predicate constants $\mathsf{q}$ that appear in $\mathsf{P}$ such that the type of
      $\mathsf{q}$ is greater than or equal to the type of $\mathsf{Q}$, it holds
      $\textit{stratum}(\mathsf{q}) < \textit{stratum}(\mathsf{p})$;
\end{enumerate}
where for every predicate constant $\mathsf{r}$,  $\textit{stratum}(\mathsf{r}) = i$ if the predicate symbol
$\mathsf{r}$ belongs to $S_i$.
\end{defi}
\begin{exa}
In the following program:
\[
\begin{array}{l}
\mbox{\tt p Q:-$\pnot$(Q a).}\\
\mbox{\tt q X:-(X$\approx$a).}
\end{array}
\]
the variable {\tt Q} is of type $\iota\rightarrow o$ and {\tt X} is of type $\iota$. The only predicate constants that appear in the program are {\tt p}, which is of type $(\iota\rightarrow o)\rightarrow o$, and {\tt q}, which is of type $\iota\rightarrow o$. Note that the type of {\tt p} is neither equal nor greater than the type of {\tt Q}, while the type of {\tt q} is the same as that of {\tt Q}.
It is straightforward to see that the program is stratified, if we choose $S_1=\{{\tt q}\}$ and $S_2=\{{\tt p}\}$. Indeed, for the first clause,
we have $stratum({\tt p}) > stratum({\tt q})$ and in the second clause there are no predicate constants or predicate variables appearing in its body.
However, if {\tt Q} and {\tt X} are as above and, moreover, {\tt Y} is of type $\iota$, it can easily be checked that the program:
\[
\begin{array}{l}
\mbox{\tt p Q:-$\pnot$(Q a).}\\
\mbox{\tt q X Y:-(X$\approx$a),(Y$\approx$a),p (q a).}
\end{array}
\]
%
%
is not stratified nor locally stratified, because if the term {\tt q a} is substituted for {\tt Q} we get
a circularity through negation. Notice that the type of {\tt q} is $\iota\rightarrow \iota \rightarrow o$ and
it is greater than the type of {\tt Q} which is $\iota \rightarrow o$.\qed
\end{exa}

Since the set of predicate constants that appear in a program $\mathsf{P}$ is finite, and since the number
of predicate constants of the program that have a greater or equal type than the type of a given predicate variable is
also finite, it follows that checking whether a given program is stratified, is decidable. Moreover, we
have the following theorem:
\begin{thm}
If $\mathsf{P}$ is stratified then it is locally stratified.
\end{thm}
\begin{proof}
Consider a decomposition $S_1,\ldots,S_r$ of the set of predicate constants of $\mathsf{P}$ such that the
requirements of Definition~\ref{stratified} are satisfied. This defines a
decomposition $S'_1,\ldots,S'_r$ of the Herbrand base of $\mathsf{P}$, as follows:
$$S'_i = \{\mathsf{A} \in U_{\mathsf{P},o} \mid \mbox{the leftmost predicate constant of $\mathsf{A}$ belongs to $S_i$}\}$$
We show that $S'_1,\ldots,S'_r$ corresponds to a local stratification of $U_{\mathsf{P},o}$.
Consider a clause in $\mathsf{P}$ of the form  $\mathsf{H}' \leftarrow \mathsf{A}'_1,\ldots,\mathsf{A}'_m,\pnot \mathsf{B}'_1,\ldots,\pnot \mathsf{B}'_n$
and let $\mathsf{H} \leftarrow \mathsf{A}_1,\ldots,\mathsf{A}_m,\pnot \mathsf{B}_1,\ldots,\pnot \mathsf{B}_n$
be one of its ground instances. Let $\mathsf{p}$ be the predicate constant of $\mathsf{H}$ (and $\mathsf{H}'$).
Consider any $\mathsf{A}'_i$. If $\mathsf{A}'_i$ starts with a predicate constant, say $\mathsf{q}_i$,
by Definition~\ref{stratified}, it is $\textit{stratum}(\mathsf{p}) \geq \textit{stratum}(\mathsf{q}_i)$.
By the definition of the local stratification decomposition we gave above, it is
$\textit{stratum}(\mathsf{H}) = \textit{stratum}(\mathsf{p})$ and
$\textit{stratum}(\mathsf{A}_i) = \textit{stratum}(\mathsf{q}_i)$,
and therefore $\textit{stratum}(\mathsf{H})\geq \textit{stratum}(\mathsf{A}_i)$.
If $\mathsf{A}'_i$ starts with a predicate variable, say $\mathsf{Q}$, then
$\mathsf{Q}$ has been substituted in $\mathsf{A}'_i$ with a term starting with a
predicate constant, say $\mathsf{q}_i$, that has a type greater than or equal
to that of $\mathsf{Q}$. By Definition~\ref{stratified}, it is
$\textit{stratum}(\mathsf{p}) \geq \textit{stratum}(\mathsf{q}_i)$ and by the
definition of the local stratification decomposition we gave above, it is
$\textit{stratum}(\mathsf{H}) = \textit{stratum}(\mathsf{p})$ and
$\textit{stratum}(\mathsf{A}_i) = \textit{stratum}(\mathsf{q}_i)$,
and therefore $\textit{stratum}(\mathsf{H})\geq \textit{stratum}(\mathsf{A}_i)$.
The justification for the case of negative literals, is similar and omitted.
\end{proof}

\section{Non-Extensionality of the Stable Models}\label{stable}
It is natural to wonder whether the more traditional approaches to the semantics of negation
in logic programming, also lead to extensional semantics when applied to higher-order logic
programs with negation under the framework of~\cite{Bezem99,Bezem01}. The two most widely known
such approaches are the {\em stable model} semantics~\cite{GL88} and the {\em well-founded}~\cite{GelderRS91}
one. It was recently demonstrated~\cite{RS17} that the well-founded approach does not in general lead to 
an extensional well-founded model when applied to higher-order programs. In this section we demonstrate 
that the stable model semantics also fails to give extensional models in the general case.
We believe that these two negative results reveal the importance of the use of the infinite-valued
approach for obtaining extensionality. 

The stable model semantics, in its original form introduced in~\cite{GL88}, is applied on the
ground instantiation of a given first-order logic program with negation, which is a (possibly infinite)
propositional program. In this respect, it is not hard to adapt it to apply to the framework
of~\cite{Bezem99,Bezem01}, which is based on the ground instantiation of a given higher-order
program (which is again a possibly infinite propositional program). For reasons of completeness,
we include the definition of stable models~\cite{GL88} (see also~\cite{Lif08}).
\begin{defi}
Let $\mathsf{P}$ be a propositional program and let $I$ be a set of propositional variables.
The {\em reduct} $\mathsf{P}_I$ of $\mathsf{P}$ with respect to $I$,  is the set of rules
without negation that can be obtained from $\mathsf{P}$ by first dropping every rule of
$\mathsf{P}$ that contains a negative literal $\pnot {\tt p}$ in its body such that
${\tt p} \in I$, and then dropping all negative literals from the bodies of all the remaining
rules. The set $I$ is called a {\em stable model} of $\mathsf{P}$ if $I$ coincides with the
least model of~$\mathsf{P}_I$.
\end{defi}

To demonstrate the non-extensionality of the stable models approach in the case of higher-order
programs, it suffices to find a program that produces non-extensional stable models. The following
very simple example does exactly this.
\begin{exa}
Consider the higher-order program:
\[
\begin{array}{l}
\mbox{\tt r(Q):-$\pnot$s(Q).}\\
\mbox{\tt s(Q):-$\pnot$r(Q).}\\
\mbox{\tt q(a).}\\
\mbox{\tt p(a).}
\end{array}
\]
where the predicate variable {\tt Q} of the first and second clause is of type $\iota\rightarrow o$. We at first take the ground instantiation of the above program:
\[
\begin{array}{l}
\mbox{\tt r(p):-$\pnot$s(p).}\\
\mbox{\tt r(q):-$\pnot$s(q).}\\
\mbox{\tt s(p):-$\pnot$r(p).}\\
\mbox{\tt s(q):-$\pnot$r(q).}\\
\mbox{\tt q(a).}\\
\mbox{\tt p(a).}
\end{array}
\]
Consider now the Herbrand interpretation $M=\{\mbox{\tt p(a)},\mbox{\tt q(a)},\mbox{\tt s(p)},\mbox{\tt r(q)}\}$.
One can easily check that $M$ is a model of the ground instantiation of the program.
However, the above model is not extensional: since {\tt p} and {\tt q} are extensionally equal,
the atoms {\tt s(q)} and {\tt r(p)} should also belong to $M$ in order to
ensure extensionality.  It  remains to show that $M$ is
also a stable model. Consider the reduct of the above program based
on $M$:
\[
\begin{array}{l}
\mbox{\tt r(q).}\\
\mbox{\tt s(p).}\\
\mbox{\tt q(a).}\\
\mbox{\tt p(a).}
\end{array}
\]
Obviously, the least model of the reduct is the interpretation $M$, and therefore
$M$ is a stable model of the initial program. In other words, we have found a
program with a non-extensional stable model.\qed
\end{exa}
Continuing the discussion on the above example, one can easily verify that the
above program also has two extensional stable models, namely
$M_1=\{\mbox{\tt p(a)},\mbox{\tt q(a)},\mbox{\tt s(p)},\mbox{\tt s(q)}\}$
and $M_2=\{\mbox{\tt p(a)},\mbox{\tt q(a)},\mbox{\tt r(p)},\mbox{\tt r(q)}\}$.
This creates the hope that we could somehow adapt the standard stable model
construction procedure in order to produce only extensional stable models.
The following example makes this hope vanish.
\begin{exa}
Consider the program:
\[
\begin{array}{l}
\mbox{\tt r(Q):-$\pnot$s(Q),$\pnot$r(p).}\\
\mbox{\tt s(Q):-$\pnot$r(Q),$\pnot$s(q).}\\
\mbox{\tt q(a).}\\
\mbox{\tt p(a).}
\end{array}
\]
where, in the first two clauses, {\tt Q} is of type $\iota\rightarrow o$. The ground instantiation of the program is the following:
\[
\begin{array}{l}
\mbox{\tt r(p):-$\pnot$s(p),$\pnot$r(p).}\\
\mbox{\tt r(q):-$\pnot$s(q),$\pnot$r(p).}\\
\mbox{\tt s(p):-$\pnot$r(p),$\pnot$s(q).}\\
\mbox{\tt s(q):-$\pnot$r(q),$\pnot$s(q).}\\
\mbox{\tt q(a).}\\
\mbox{\tt p(a).}
\end{array}
\]
This program has the non-extensional stable model $M=\{\mbox{\tt p(a)},\mbox{\tt q(a)},\mbox{\tt s(p)},\mbox{\tt r(q)}\}$.
However, it has {\em no} extensional stable models: there are four possible extensional interpretations
that are potential candidates, namely $M_1=\{\mbox{\tt p(a)},\mbox{\tt q(a)}\}$,
$M_2=\{\mbox{\tt p(a)},\mbox{\tt q(a)},\mbox{\tt r(p)},\mbox{\tt r(q)}\}$,
$M_3=\{\mbox{\tt p(a)},\mbox{\tt q(a)},\mbox{\tt s(p)},\mbox{\tt s(q)}\}$,
and $M_4=\{\mbox{\tt p(a)},\mbox{\tt q(a)},\mbox{\tt s(p)},\mbox{\tt s(q)},\mbox{\tt r(p)},\mbox{\tt r(q)}\}$;
one can easily verify that none of these interpretations is a stable model of the
ground instantiation of the program. The conclusion is that there exist higher-order logic
programs with negation which have only non-extensional stable models! \qed
\end{exa}

The above examples seem to suggest that the extensional approach of~\cite{Bezem99,Bezem01}
is incompatible with the stable model semantics. Possibly this behaviour can be explained
by an inherent characteristic of the stable model semantics, which appears even in the case
of classical logic programs with negation. Consider for example the simple propositional program:
\[
\begin{array}{l}
\mbox{\tt p:-$\pnot$q.}\\
\mbox{\tt q:-$\pnot$p.}
\end{array}
\]
The above program has two stable models, namely $\{\mbox{\tt p}\}$ and $\{\mbox{\tt q}\}$.
In the first stable model, {\tt p} is true and {\tt q} is false despite the fact that the
program is completely symmetric and there is no apparent reason to prefer {\tt p} over {\tt q};
a similar remark applies to the second stable model. It is possible that it is this characteristic
of stable models (i.e., the resolution of negative circularities by making specific choices that
are not necessarily symmetric) that leads to their non-extensionality.

\section{Connections to the Research of Zolt\'{a}n \'{E}sik}\label{conclusions}
The work reported in this paper is closely connected with research conducted by Zolt\'{a}n \'{E}sik
in collaboration with the first author (Panos Rondogiannis). In this section we briefly describe the
roots of this joint collaboration which unfortunately was abruptly interrupted by the untimely loss
of Zolt\'{a}n.

In 2005, the first author together with Bill Wadge proposed~\cite{RondogiannisW05} the infinite-valued
semantics for logic programs with negation (an overview of this work is given in Section~\ref{infinite_valued}
of the present paper). In 2013, the first author together with Zolt\'{a}n \'{E}sik started a collaboration
supported by a Greek-Hungarian Scientific Collaboration Program with title ``Extensions and Applications
of Fixed Point Theory for Non-Monotonic Formalisms''. The purpose of the program was to create
an abstract fixed point theory based on the infinite-valued approach, namely a theory that
would not only be applicable to logic programs but also to other non-monotonic formalisms.
This abstract theory was successfully developed and is described in detail in~\cite{ER14}.
As an application of these results, the first extensional semantics for higher-order
logic programs with negation was developed in~\cite{CharalambidisER14}. Another application 
of this new theory to the area of non-monotonic formal grammars was proposed in~\cite{ER14a}. 
Moreover, Zolt\'{a}n himself further investigated the foundations and the properties of the 
infinite-valued approach~\cite{Esik15}, highlighting some of its desirable characteristics.

The work reported in the present paper is an alternative approach to the semantics of higher-order
logic programs with negation developed in collaboration with Zolt\'{a}n in~\cite{CharalambidisER14}.
The two approaches both use the infinite-valued approach but are radically different otherwise.
In particular, the present approach heavily relies on the ground instantiation of the source program,
while the approach of~\cite{CharalambidisER14} operates directly on the source program and extensively
uses domain theoretic constructions. It is easy to find a program where our approach gives
a different denotation from that of~\cite{CharalambidisER14}. Actually, the two approaches
differ even for positive programs. The example given below is borrowed from~\cite{CharalambidisRS15}
where it was used to demonstrate the differences between the approach of~\cite{Bezem99,Bezem01}
versus that of~\cite{CharalambidisHRW13} (which applies to positive programs but has
similarities to~\cite{CharalambidisER14} because it is also based on domain theory).
\begin{exa}\label{non-equivalent}
Consider the following program:
\[
\begin{array}{l}
\mbox{\tt p(a):-Q(a).}
\end{array}
\]
where the predicate variable {\tt Q} is of type $\iota\rightarrow o$. Under the semantics developed in this paper (which generalizes that of~\cite{Bezem99,Bezem01}),
the above program intuitively states that {\tt p} is true of {\tt a} if there exists
{\em a predicate that is defined in the program} that is true of {\tt a}. If we take the
ground instantiation of the above program:
\[
\begin{array}{l}
\mbox{\tt p(a):-p(a).}
\end{array}
\]
and compute the least model of the instantiation, we find the least model of
the program which assigns to the atom {\tt p(a)} the truth value $F_0$.

Under the semantics of~\cite{CharalambidisER14} the atom {\tt p(a)} is true
in the minimum Herbrand model of the initial program. This is because
in this semantics, the initial program reads (intuitively speaking) as follows:
``{\tt p(a)} is true if there exists a {\em relation} that is true of {\tt a}''; actually,
there exists one such relation, namely one in which the constant {\tt a} has the value $T_0$.
This discrepancy between the semantics of~\cite{Bezem99,Bezem01} and that of~\cite{CharalambidisER14}
is due to the fact that the latter is based on (infinite-valued) {\em sets} and not on the syntactic
entities that appear  in the program.\qed
\end{exa}
Despite the above difference, there certainly exists common ground between
the two techniques. As it is demonstrated in~\cite{CharalambidisRS15}, there exists
a large and useful class of positive programs for which the approach
of~\cite{Bezem99,Bezem01} coincides with that of~\cite{CharalambidisHRW13}.
Intuitively speaking, this class contains all positive programs that do not contain
{\em existential} variables in the bodies of clauses (like $\mathsf{Q}$ in the
above example). We believe that such a result also holds for programs with negation.
More specifically, we conjecture that there exists a large fragment of $\mathcal{H}$
for which the semantics of~\cite{CharalambidisER14} coincides with the one given in
the present paper. Establishing such a relationship between the two semantics, will
lead to a better understanding of extensional higher-order logic programming.

\bibliographystyle{alpha}

\end{document}